\pdfoutput=1
\documentclass[12pt,letterpaper,reqno]{article}
\usepackage[markup=underlined]{changes}
\usepackage{etex}
\usepackage{graphicx}
\usepackage{amsmath}
\usepackage{amssymb}
\usepackage{amsfonts}
\usepackage{amsthm}
\usepackage{titlesec}
\usepackage{comment}
\usepackage{graphics,epsfig,verbatim,bm,latexsym,url,amsbsy}
\usepackage{rotating}
\usepackage[authoryear,round]{natbib}
\usepackage{float}
\usepackage[position=bottom]{subfig}
\usepackage{mathrsfs}
\usepackage{multirow}
\usepackage{array}
\usepackage{bigints}
\usepackage{bbm}
\usepackage[ruled,vlined]{algorithm2e}
\usepackage{soul}
\usepackage{mathtools}
\usepackage{enumitem}
\usepackage{multirow,array}
\usepackage{blkarray}
\usepackage{kpfonts}
\usepackage{setspace}

\theoremstyle{definition} \newtheorem{example}{Example}
\theoremstyle{definition} 
\theoremstyle{definition} \newtheorem{corollary}{Corollary}
\theoremstyle{definition} 
\theoremstyle{definition} 
\theoremstyle{definition} 
\theoremstyle{definition} 
\theoremstyle{definition} \newtheorem{lemma}{Lemma}
\theoremstyle{definition} \newtheorem{theorem}{Theorem}
\theoremstyle{definition} \newtheorem{assumption}{Assumption}
\theoremstyle{definition} 
\theoremstyle{definition} 
\theoremstyle{definition} 
\theoremstyle{definition}\newtheorem{proposition}{Proposition}
\theoremstyle{definition} 
\theoremstyle{definition} 
\theoremstyle{definition} 
\theoremstyle{definition} 
\theoremstyle{definition} 
\theoremstyle{definition} 
\theoremstyle{definition} 

\newcommand{\Ex}{\mathbb{E}}

\renewcommand{\Pr}{\mathbb{P}}

\DeclarePairedDelimiterX{\inp}[2]{\langle}{\rangle}{#1, #2}

\usepackage[letterpaper,left=1.25in,right=1.25in,top=1.25in,bottom=1.25in]{geometry}
\selectfont % Adjusted to ensure ≤32 lines/page
\titleformat{\section}
    {\bfseries\center}
    {\thesection}
    {0.5em}
    {}
    []
    
\titleformat{\subsection}[runin]
    {\normalfont\bfseries}
    {\thesubsection}
    {0.5em}
    {\addperiod}
    []
\newcommand{\addperiod}[1]{#1.}

\usepackage{hyperref}
\hypersetup{
    colorlinks=true,
    linkcolor=red!75!black,
    citecolor=blue!75!black,
    filecolor=magenta,      
    urlcolor=cyan,
}

\title{Inertial Coordination Games}

\author{\makebox[.25\linewidth]{{Andrew Koh}\thanks{MIT Department of Economics; email: \protect\texttt{ajkoh@mit.edu}}}\\{MIT} 
\and 
\makebox[.25\linewidth]{Ricky Li\thanks{MIT Department of Economics; email: \protect\texttt{rickyli@mit.edu}}} \\ {MIT} 
\and \makebox[.25\linewidth]{Kei Uzui\thanks{MIT Department of Economics; email: \protect\texttt{kuzui@mit.edu} \\~\\
We are especially grateful to Drew Fudenberg and Stephen Morris for insightful suggestions. We also thank Isaiah Andrews, Ricardo Caballero, Roberto Corrao, Tuval Danenberg, Daniel Luo, Parag Pathak, Harry Pei, David Romer, Satoru Takahashi, Iv\'an Werning, Alex Wolitzky, Muhamet Yildiz, participants at MIT Theory and Macro Lunches, and anonymous referees for excellent comments. An extended abstract of this paper appeared in the Proceedings of the 26th ACM Conference on Economics and Computation (EC'25).
}} \\
{MIT}}

\date{This version: August 12, 2025}

\begin{document}

\maketitle
\thispagestyle{empty}

\begin{abstract}
We analyze \textit{inertial coordination games}: dynamic coordination games with an endogenously changing state that depends on (i) a persistent fundamental players privately learn about over time; and (ii) past play. The speed of learning determines long-run equilibrium dynamics: the risk-dominant action is played in the limit if and only if learning is \emph{slow} such that posterior precisions grow sub-quadratically. This generalizes results from static global games and endows them with a learning foundation. Conversely, when learning is \emph{fast} such that posterior precisions grow super-quadratically, shocks can propagate and generate self-fulfilling spirals.%Whenever the risk dominant equilibrium is selected, the path of play undergoes a sudden phase transition when signals are precise, but a gradual transition when signals are noisy. 
\end{abstract}

\clearpage

\clearpage
\setcounter{page}{1}

\section{Introduction} 
Coordination games are at the heart of many dynamic economic phenomena. In such games, shocks to the fundamentals can sometimes \emph{propagate}: as traders attack a currency peg and the central bank expands reserves to maintain the peg, more traders are incentivized to take advantage of this weakness which, in turn, further drains the central bank's reserves, enticing even more attacks. But shocks can also \emph{fizzle out}: traders might learn quickly that the central bank's balance sheet is healthy so that the proportion of the population who are (incorrectly) pessimistic thins out rapidly. In this case, the downward spiral of an initial attack begetting further attacks eventually loses steam.\footnote{Similar forces operate with bank runs (via bank liquidity), networked products (via product improvements), firm pricing (via sticky consumers), and deflationary spirals (via price and demand downturns).} When do shocks propagate and when do they fizzle out? 

%These dynamics in which past attacks endogenously influences the present state which, in turn, shapes present incentives are difficult to square with `standard' models of self-fulfilling crises. 

To answer these questions, we study \emph{inertial coordination games}: dynamic coordination games in which players repeatedly decide whether or not to take a risky action (e.g., attack a currency). The payoff of this risky action depends on both (i) a persistent fundamental component; and (ii) an endogeneous component which depends on past play (e.g., if more players have attacked a currency in the past, this depletes current reserves). Players privately learn over time and form inferences about the current state, which shapes their current play. This, in turn, propagates into the future by shaping the evolution of future states. Thus, our environment features two interconnected forces: \emph{private learning}---in which dispersion in beliefs about fundamentals across the population shrinks with time---and \emph{intertemporal contagion}, in which higher current play increases future states which incentivizes higher future play.

Our main result (Theorem \ref{thrm:char}) develops a tight connection between the speed of learning and limit play: the risk-dominant action is played in the limit \emph{if and only if} posterior precisions grow sub-quadratically. Sub-quadratic learning includes canonical learning environments such as (i) independent and identically distributed signals; and (ii) the receipt of a one-time signal. On the other hand, super-quadratic learning includes social learning such that, in each period, agents are randomly matched and sample their partner's information sets. The tight connection we develop between learning speeds and limit play has simple but sharp economic implications. 
Slow (sub-quadratic) learning implies \emph{history independence}: the initial shock to aggregate play as well the fine details of what agents learn in the interim is irrelevant for limit play. Conversely, when learning speeds are fast (super-quadratic), this induces \emph{history dependence} such that initial shocks can be self-fulfilling: whether shocks propagate or not is \emph{jointly sensitive} to both the fundamentals and the magnitude of the shock.\footnote{See Section I of \cite{krugman1991history} for a discussion of history-dependence and expectations in economics.} 
%Whether fast or slow learning is beneficial for the regime\footnote{In an ex-ante sense, before the state or shocks are realized.} is ultimately dependent on how prone the regime is to shocks.

Our model of inertial coordination games imposes a lag between players' actions and their effect on the time-varying state. Hence, when a player contemplates whether or not to play the risky action, his belief about the current state depends on his belief about past play. But that, in turn, depends on his belief about other players' past beliefs which themselves are about play yet further back into the past. Thus, the path of past play (and beliefs)---even of the distant past---multiply forward to shape incentives in the present.\footnote{Note that our environment is, in a loose sense, the opposite of coordination games with switching frictions \citep{frankel2000resolving}. In those models, switching frictions lead players to form beliefs about future play and hence future beliefs; by contrast, in our model, inertia leads players to form beliefs about past play and hence past beliefs.} An important by-product of inertia is that it offers a simple yet novel means of obtaining unique predictions in coordination environments. 

At the same time, our model is motivated by---and consistent with---recent work documenting empirical facts about bank runs. In their empirical analysis of minute-by-minute depositor-level data from an Indian bank, \cite*{iyer2012understanding} find that depositors’ beliefs about “the fraction of other depositors... that have run until time $t-1$" help determine whether they run at time $t$ (eg, inertia). More generally, \cite*{correia2024failing} use historical data from American commercial banks to assess the two prevailing explanations of bank failure in this literature: the ``\textit{bank runs view}" in which failures are due to self-fulfilling runs on solvent banks, and the ``\textit{solvency view}" in which insolvency due to poor fundamentals is the primary reason underlying bank failure. In addition to matching \cite*{correia2024failing}'s finding that \emph{``depositors tend to be slow to react to an increased risk of bank failure''} (eg, the inertia), our main result suggests that which of the above two explanations is best is characterized by the speed of learning among depositors: slow learning implies that outcomes are solely determined by fundamentals as in the ``solvency view," while fast learning implies that self-fulfilling runs may occur even for banks with strong fundamentals as in the ``bank runs view." Hence, one of our contributions is to offer a simple dynamic model which reconciles empirical regularities about `real world' coordination environments with the more abstract logic of learning and higher-order beliefs in coordination games. Further, our characterization of the connection between speed and limit play is invariant to the degree of inertia (duration of lag) and, as inertia vanishes, the equilibrium converges to an equilibrium of the model with purely contemporaneous incentives.\footnote{These claims are made precise in Appendix \ref{app:small lags}.}

In what sense does our model relate to the extant literature on coordination games? It turns out that static global games \citep{carlsson1993global} are recovered \emph{exactly} as a special case of our model in which players only receive a single signal at time $t = 1$: since this learning process is sub-quadratic (posterior  precisions are unchanged after $t = 1$), the risk-dominant profile is---as in global games---played in the limit. This connection arises from observing that when there is no interim learning, best-responding to expected past play along each time step in our model is equivalent to the interim elimination of strictly dominated strategies in the static game. As we make precise in Appendix \ref{app:microf}, this uncovers a formal connection between ``contagion'' arguments over type spaces (which encode hierarchies of higher-order beliefs) and ``contagion'' over time (where the current state depends on past play). Thus, our results generalize the predictions of static global games to sub-quadratic learning and offer a simple dynamic learning foundation for the `vanishing noise limit' commonly invoked in that literature.

%We further characterize all possible equilibria: any monotone selection between initial play-dominance (under common-knowledge of the fundamental state, players best-respond to the conjecture that the proportion of players choosing the risky action is given by initial play) and risk-dominance can be implemented via some appropriate learning rate. Thus, our results offer a simple and unified analysis for when our model's dynamics exhibit \emph{history-dependence} such that initial shocks can be self-fulfilling,\footnote{See Section I of \cite{krugman1991history} for an excellent discussion of history-dependence and expectations in economics.} and when it exhibits \emph{history-independence} so that limit play is independent of the initial shock. 

Finally, our model makes sharp qualitative predictions about the \emph{path of play} (Proposition \ref{prop: dynamics}) whenever learning is slow enough such that the risk-dominant profile is played in the limit. In particular, transition dynamics toward limit play depend starkly on the precision of signals. When signals are precise, aggregate play exhibits a \emph{sudden phase transition} from approximately all players playing the non-risk-dominant action (potentially for a long while) to approximately all players playing the risk-dominant action thereafter. Conversely, when signals are noisy, aggregate play exhibits a \emph{gradual transition} as the measure of players playing risk-dominant action increases gradually. These differing transition regimes are driven by differences in time-varying heterogeneity in beliefs. The importance of heterogeneity has been recognized as a crucial determinant of equilibrium uniqueness in coordination games \citep{morris2006heterogeneity}. Our results complement this insight by highlighting that the speed of learning---which drives how the cross-sectional distribution of beliefs vary over time---is also a crucial determinant of the \emph{path} of play. It also suggests that ``spikes'' in the time-series of aggregate actions (e.g., withdrawal decisions or adoption of networked goods) can be consistent with transition to limit equilibrium play and need not necessarily correspond to ``equilibrium shifts" \citep{chwe2013rational,morris2019crises}. 

\textbf{Related literature.} Our paper directly relates to the large literature on global games. A broad takeaway from the literature on static global games \citep*[among others]{carlsson1993global,kajii1997robustness,frankel2003equilibrium} is that perturbing common knowledge by introducing a vanishing amount of private uncertainty often selects the risk dominant equilibrium. This approach has been widely deployed in a variety of economic applications to select among multiple equilibria (see, e.g., \cite{morris1998unique,goldstein2005demand} among many others). Our model nests the canonical global games model of \cite*{carlsson1993global} as a special case with a private signal at $t = 1$ and no learning thereafter. The ``contagion argument'' in global games where players iteratively eliminate interim dominated strategies is sometimes viewed metaphorically or as a convenient way to select among equilibria.\footnote{See, \cite{morris2002coordination} and Chapter 2 of \cite{spiegler2024curious}.} Our simple dynamic environment with an endogeneously changing state thus offers an explicit dynamic foundation for the ``contagion'' logic of global games consistent with recently documented facts on bank runs \citep*{iyer2012understanding,correia2024failing}. Moreover, we substantially generalize the prediction of risk-dominant selection to the regime of sub-quadratic learning.

Another strand of the literature has investigated equilibrium selection from a non-Bayesian perspective, in particular under exogenous mutations or limited memory (\cite{fudenberg1992evolutionary}, \cite*{kandori1993learning},\cite{young1993evolution}, \cite*{bergin1996evolution}, \cite*{steiner2008contagion}, \cite*{block2019learning}). In contrast, we propose a model of private Bayesian learning in which the speed of learning and initial conditions jointly mediate limit play. Our finding that super-quadratic learning rates lead limit play to depend on initial conditions is related to the ``hysteresis equlibrium" in macroeconomic settings modelled as coordination games discussed in \cite*{cooper1994equilibrium,krugman1991history} and, more recently, in \cite{morris2019crises}. Our results develop a tight characterization of when hysteresis plays a role for limit outcomes, and when initial conditions wash out.\footnote{See also \cite{chen2016falling} and \cite*{kozlowski2020tail} for a different mechanism of how model uncertainty and rare events interact and give rise to persistence.}

Our work also relates closely to dynamic coordination games \citep*[among others]{angeletos2007dynamic,dasgupta2007coordination,dasgupta2012dynamic}. In much of this literature, coordination is contemporaneous i.e., time-$t$ play matters for time-$t$ incentives. An important exception is the paper of \cite{mathevet2013tractable} who study finite-time horizon dynamic global games with non-contemporaneous coordination motives where the terminal outcome depends on both the state and path of play.\footnote{See also the elegant paper of \cite*{dasgupta2012dynamic}.} Our models are non-nested and deliver complementary insights. In our model, the state is time-varying and endogenous to past play so players' payoffs depend on interim play which in turn drives the state-evolution. In \cite{mathevet2013tractable}, ex-post payoffs depend on a binary outcome which is determined by a persistent fundamental and the path of aggregate play. This mapping from paths to a binary outcome allows \cite{mathevet2013tractable} to tractably handle a wide class of payoffs. By contrast, our payoffs are more specific, but evolve endogeneously and take on a continuum of values. We exploit the structure of our environment to derive sharper implications of learning speeds on limit outcomes.\footnote{The main result of \cite{mathevet2013tractable} is that aggregate play at the critical state at which outcomes hinge only depends on payoffs; by contrast, we show that limit aggregate play at any state depends sharply on the speed of learning. Note also that \cite{mathevet2013tractable}'s definition of ``fast learning'' takes the precision of each signal arbitrarily large, which is sufficient for equilibrium existence. This is different from our notion which requires the (weaker) condition of super-quadratic growth in posterior precisions, and where limit equilibrium play always converges.}

Our model can be viewed as the conceptual opposite of models which impose switching frictions to obtain unique equilibrium predictions \citep{frankel2000resolving}.\footnote{See also \cite*{frankel2003equilibrium} and more recent work in dynamic information design by \cite*{koh2024informational}. \cite{angeletos2016incomplete} survey work in macroeconomics applying switching frictions to obtain equilibrium selection.} 
A key feature of our model is that players can frictionlessly change their actions, but there is a (potentially small) lag between players' actions and transmission to the state. This lag links incentives intertemporally: when a player contemplates whether or not to play the risky action, he must form beliefs about \emph{past play} and so beliefs about even the distant past can multiply forward to shape present incentives. Thus, our model can be loosely viewed as the conceptual counterpart of models with switching frictions which require that players form beliefs about \emph{future play} \citep{frankel2000resolving}. Unlike such models which predict that the risk-dominant action will be selected, our model offers a more nuanced unique prediction which hinges on the speed of private learning. 

%For instance, the important work of \cite*{angeletos2007dynamic} also features private learning. Different from their model, coordination motives in ours are non-contemporaneous and arises through a time-varying endogenous state.\footnote{Intertemporal linkages in \cite*{angeletos2007dynamic} are implicitly driven by observational learning that the state has thus far survived.} 

Finally, our results are connected to the literature on learning in coordination games. An interpretation of our model is of best-response dynamics with Bayesian learning about a fundamental state or past play. To our knowledge, this is novel to the extant literature on learning in games which assumes that past play is perfectly observed and there is no fundamental uncertainty.\footnote{See \cite{fudenberg1998theory} for a survey.} Our paper is related to an important paper of \cite{crawford1995adaptive} who studies adaptive dynamics in complete information coordination games by modelling player strategies as a linear adjustment rule to be estimated; this is similar to the interpretation of players in our model as best-responding to past play. Thus, a contribution of our paper is to analyze how Bayesian learning about fundamental uncertainty interacts with adaptive learning about strategic uncertainty \citep{camerer1999experience}.\footnote{\cite{camerer1999experience} develop a model of learning in complete information games which features both modes of learning; by contrast, in the best-response interpretation of our model, different modes of learning correspond to different kinds of uncertainty: `fundamental uncertainty' is resolved via Bayesian learning but `strategic uncertainty' is resolved via adaptive or reinforcement learning.
}

%It also relates to a literature on best-response dynamics because the literal interpretation of out model is that agents are best-responding to past play (which they observe with some noise). 
 
% \textcolor{blue}{[Ricky: I'm still not sure how to reconcile our dynamics being fully deterministic with a fictitious play interpretation.]}

\section{Model}
\textbf{Payoffs.} There is a unit measure of agents indexed by $i \in [0,1]$.\footnote{Endow $[0,1]$ with the Lebesgue extension constructed in \cite*{sun2009individual} so that the continuum law of large numbers holds.} 
Time is discrete and indexed by $t \in \mathcal{T} = \{1,2,\ldots\}$. At each time $t \in \mathcal{T}$, each agent $i$ chooses between a risky action $(a_{it} = 1)$ or a safe action $(a_{it} = 0)$. Let $\lambda_t$ denote the measure of agents who play the risky action at time $t$.

There is an unknown and persistent fundamental $\theta \in \mathbb{R}$. For each $t \in \mathcal{T}$, let $\theta_t$ be the state at time $t$ and for each $t>1$,\footnote{For any $\alpha \in (0,1)$, analogs of our results apply for $\theta_{t} = \alpha \theta + (1-\alpha) \lambda_{t-1}$; see Appendix \ref{appendix:payoffextension}.} 
\[
    \theta_t := \underbrace{\theta}_{\text{Persistent}} + \underbrace{\lambda_{t-1}}_{\text{Endogenous}}  \quad \text{with initial play $\lambda_0 \in [0,1]$.}
\]
Our law of motion for the endogenous state is simple. Nonetheless, we think it well-approximates a range of environments such as currency attacks or bank runs in which the failure probability depends on a single sufficient statistic: present reserves. This, in turn, is simply the sum of its initial reserves less the `current' short/withdrawal positions. For a currency attack where the risky action ($a = 1$) is to attack, initial reserves are $-\theta$ and the current short position is $\lambda_{t-1}$; for a bank run the risky action ($a = 1$) is to deposit, initial reserves are $\theta$, and the current withdrawal position is $1-\lambda_{t-1}$.\footnote{As is standard in global games, we have assumed that $\theta$ is supported on the real line, but it can be appropriately shifted in space and truncated so that the sign fits the economic phenomena in question.}

Each agent $i$'s payoffs at time $t$ are:
\[
u(a_{it},\theta_t):=\begin{cases} 
\theta_t & \text{if } a_{it}=1 \\
c & \text{if } a_{it}=0 \\
        \end{cases}
\]
In the main text, we assume that $\lambda_0$ is common knowledge but interim payoffs as well as $(\lambda_t)_{t \geq 1}$ are unobserved and inferred via beliefs about the state $\theta$.\footnote{An identical assumption is made in the context of dynamic global games by \cite*{angeletos2007dynamic}.} As Appendix \ref{app:het-unknown} shows, the assumption that $\lambda_0$ is common knowledge may be relaxed. Our model is equivalent to one in which agents observe noisy signals of past play $(\lambda_t)_{t \geq 1}$, or both signals about the state as well as past play---we formalize this  an equivalence in Appendix \ref{appendix:pastplay} and discuss this below.

\textbf{Interpreting initial play.} We interpret $\lambda_0 \in [0,1]$ as a shock. This is equivalent to a one-off and transitory shock to the cost at time $1$.\footnote{More explicitly, suppose there is zero initial play but there is an initial shock to costs of playing the risky action. Costs are potentially time-varying $c_1,c_2,\ldots$ with $c:= c_2 = c_3 = c_4 =\ldots$ and $c_1 \neq c$. Then, the dynamics of this model with zero initial play is equivalent to the model with $\lambda_0 = c - c_1$.} The impact of this shock is equivalent to that of non-zero initial play: positive shocks which incentivize the risky action correspond to high $\lambda_0$ and vice versa. In the language of macroeconomics, this corresponds to an unexpected shock (``MIT shock'') and we analyze the dynamics of aggregate play following the shock. In the language of fictitious play, this corresponds to a shock to initial play and we analyze how the speed of learning shapes the limit action profile.

\textbf{Timing and Information.} The timing of the game and the informational environment are as follows. All agents share a common improper uniform prior about $\theta$. At each time $t \in \mathcal{T}$, each agent $i$ receives a private signal $x_{it} \sim N(\theta,\sigma_t^2)$ that is iid across agents and independent across time. Each agent $i$ then updates to her time-$t$ posterior about $\theta$ conditional on the natural filtration $\mathcal{F}_{it}$ generated by past signals $\{x_{is}\}_{s=1}^t$. From standard Bayesian updating, this posterior is also Gaussian: $\theta|\mathcal{F}_{it} \sim N(\mu_{it},\eta_t^2)$, where $\mu_{it}:=\eta_t^2 \sum_{s=1}^t \sigma_s^{-2}x_{is}$ and $\eta_t^2:=(\sum_{s=1}^t \sigma_s^{-2})^{-1}$. Finally, each agent $i$ chooses $a_{it} \in \{0,1\}$ to maximize her expected utility. Note that since agents are atomless, they cannot influence future aggregate play, so it is without loss to assume they are myopic.\footnote{For example, we could redefine agent $i$'s payoffs as $\max_{(a_{it})_{t\geq1}} \sum_{t\geq 1} \delta^t \mathbb{E}_t[u(a_{it},\lambda_{t-1},\theta)]$ for some $\delta \in (0,1)$.}

\begin{figure}[h!]  
\centering
\captionsetup{width=0.9\linewidth}
    \caption{Timing and intertemporal linkages} \includegraphics[width=0.8\textwidth]{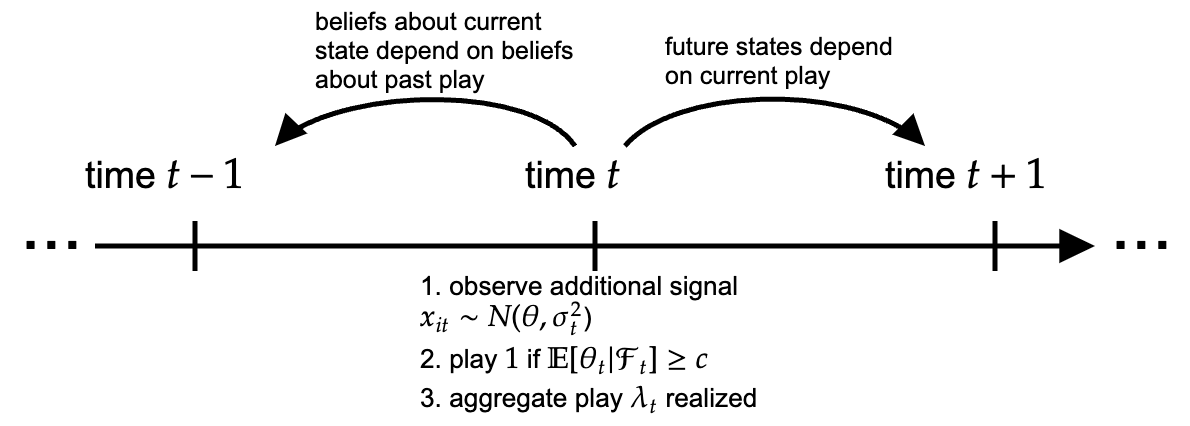}
     \label{fig:timing}
\end{figure} 

\textbf{Limits and Learning.} Let the \emph{learning process} $\Sigma:=(\sigma_t^2)_{t \in \mathcal{T}}$ be a strictly positive sequence of signal variances determining the speed of learning. Note that instead of $\Sigma$, we may equivalently directly specify a nonnegative, increasing sequence of posterior precisions $(\eta_t^{-2})_t$, since such a sequence is induced by $\sigma_t^{-2}:=\eta_t^{-2}-\eta_{t-1}^{-2}$. When convenient, we will therefore take $(\eta_t^{-2})_t$ as a primitive of our model.

We maintain the following assumption unless otherwise stated. 
\begin{assumption}[Limit learning]\label{assumption: limit learning}$\Sigma$ satisfies $\lim_{t \to \infty} \eta_t^{-2}= + \infty$.
\end{assumption}

A key object of analysis in this paper is \textit{limit aggregate play} at state $\theta$, as a function of signal variances $\Sigma$ and initial aggregate play $\lambda_0$. We denote this by $\lambda_\infty(\theta|\Sigma,\lambda_0):=\lim_{t \to \infty} \lambda_t(\theta|\Sigma,\lambda_0)$. When $(\Sigma,\lambda_0)$ is understood to be fixed, we omit dependencies. In particular, we study when limit aggregate play converges to the \textit{risk dominant (RD) action} at state $\theta$, denoted by 
\[RD(\theta):=\mathbbm{1}\big\{\theta \geq c-(1/2)\big\}.\footnote{Our notion of convergence of functions of $\theta$ is a.e. pointwise convergence.}\] 
The connection to risk-dominance is as follows: $RD(\theta)$ is the best-response for each agent if she knew $\theta$ and believed that half the population played the risky action.\footnote{See \cite{morris_shin_2003} for a survey of the connection between risk-dominance and continuum player global games.} Analogously, define the \textit{non-risk dominant} action as $NRD(\theta):=1-RD(\theta)$.

We discuss several aspects of our model: 
\begin{itemize}[leftmargin = 1em]
    \item \textbf{State transition \& non-contemporaneous incentives}:   Our model is one in which the state at time-$t$ ($\theta_t$) is shaped by play at time-$t - 1$ ($\lambda_{t-1}$). A few remarks are in order. \underline{First}, we think this captures an important element of dynamic coordination phenomena (such as the bank runs studied in \cite{iyer2012understanding} and \cite*{correia2024failing}). For instance, in currency crises the susceptibility of the current regime depends on the short position. Note also that our model exhibits path dependence---though $\theta_{t}$ depends directly only on $\lambda_{t-1}$, it also depends indirectly on the full path of past play $(\lambda_s)_{s < t-1}$.\footnote{This is broadly analogous to how an AR$(1)$ process depends indirectly on realizations of the process which are far into the past (although our setting is non-stationary). An alternate interpretation is that our model is a form of best-response dynamics. Of course, some care is required with this interpretation since past play is unobserved: players are thus best-responding to their \emph{best guess} of past play. 
    Experimental evidence in (complete information) coordination games find that the vast majority (96\%) of play constitute myopic best-responses \citep{mas2016behavioral}.} This allows our model to make rich and nuanced predictions on how the speed of learning shapes equilibrium selection. By contrast, a model with \emph{exact} contemporaneous incentives cannot capture this---the risk-dominant equilibrium is always selected.\footnote{In Appendix \ref{app:small lags} we show that the unique belief threshold is $c - 1/2$ so the risk-dominant action is always selected in the limit.} 
    \underline{Second}, our main result on learning regimes (Theorem \ref{thrm:char}) holds for vanishing inertia: if we fixed the speed of learning as a function of calendar time but shrink the spacing between time periods, the sub/super-quadratic threshold for risk/non-risk dominance continues to hold. This is in the same spirit as the `frictionless limit' of \cite{frankel2000resolving} where switching frictions vanish. \underline{Third}, as the degree of inertia vanishes, the (unique) equilibrium of our inertial coordination game converges to an equilibrium of the coordination game with contemporaneous incentives.\footnote{That is, with utility function $u(1,\theta,\lambda_{t}) = \theta + \lambda_t$. Proposition \ref{prop:smallinertia} in Appendix \ref{app:small lags} formalizes this.} This means that the speed of learning \emph{relative to} the degree of inertia is what is important for obtaining richer predictions about equilibrium outcomes. 

\item 
\textbf{Observability.} We have assumed agents do not observe aggregate past play.\footnote{Aggregate past play is an injective function of the state, so if agents were to observe this, they would immediately learn the state.} However, our model is equivalent to one in which agents receive private signals both about the state, and about aggregate past play.\footnote{For instance, see \cite{dasgupta2007coordination,trevino2020informational}.} Explicitly, suppose that agent $i \in [0,1]$, in addition to observing $x_{it} \sim N(\theta, \sigma^2_t)$, additionally observes an independent signal distributed $y_{it} \sim N(\Phi^{-1}(\lambda_{t-1}), \tau_t^2)$; this signal structure was introduced by \cite{dasgupta2007coordination} and used recently by \cite{trevino2020informational} to study financial contagion. Then, we can show by induction that equilibrium behavior is equivalent to a model in which players receive a more precise signal about $\theta$ and no information about aggregate past play; Appendix \ref{appendix:pastplay} formalizes this connection. In light of this observation, we load all learning on the state, and agents make inferences about aggregate past play conditional on their beliefs about the state.
\item 
\textbf{Continuum of players.} We follow the literature on global games by assuming that there is a continuum of atomless players. By the continuum law of large numbers, this implies that conditioned on the fundamental state $\theta$, all aggregate randomness in private  signal realizations across the population washes out. This allows us to analyze the dynamics of aggregate play cleanly as the solution to a non-linear deterministic difference equation. In Appendix \ref{appendix:finiteplayers}, we develop a finite-player version of our model with $N < +\infty$ players and show via standard concentration arguments of the empirical distribution of beliefs that limit play behaves similarly in the finite-player version of our model (see Theorem \ref{thrm:approx} in Appendix \ref{appendix:finiteplayers}). 
\end{itemize}

%\footnote{Explicitly, in this model agent $i$ observes $y_{it} \sim N(\Phi^{-1}(\lambda_{t-1}), \tau_t^2)$ at time $t$ as used in \cite{dasgupta2007coordination,trevino2020informational}.} However, such signals are equivalent to simply receiving a more precise signal about $\theta$ and no information about aggregate past play. In light of this observation, we load all learning on the state, and agents make inferences about aggregate past play conditional on their beliefs about the state.$\hfill \diamondsuit$

\section{Equilibrium law of motion}
We begin by analyzing the joint evolution of aggregate beliefs and aggregate play for the learning process and initial profile of play $(\Sigma,\lambda_0)$. Appendix \ref{appendix:proofs} contains the relevant proofs.

\textbf{Dynamics of threshold beliefs.} Recall that agent $i$'s time-$t$ posterior is $N(\mu_{it},\eta_t^2)$. Agent $i$ chooses the risky action $a_{it}=1$ at time $t$ if and only if $\mu_{it}\geq \mu_t^*$, where the \textit{belief threshold} $\mu_t^*$ is implicitly defined by the indifference condition
\[
\mu^*_t + \Ex_{\theta \sim N(\mu^*_t,\eta^2_t)} \Big[\lambda_{t-1}(\theta)\Big] = c.
\]
Let $\Phi$ be the standard Gaussian CDF. Expanding this indifference condition yields the law of motion for belief thresholds $(\mu^*_t)_t$, which can be written explicitly as the nonlinear difference equation
\[
\mu^*_t = \mu^*_{t-1} + \Big(\sqrt{\eta^2_{t-1} + \eta^2_t} \Big) \cdot \Phi^{-1}(c - \mu^*_{t}) \quad \text{for all $t \geq 1$}
\tag{$\mu^*$-LOM}
\label{main-text-mu-lom}
\]
with boundary $\mu_1^*:=c-\lambda_0$. Note that the dynamics of $(\mu_t^*)_t$ are (i) monotone in $t$; (ii) do not depend on the state $\theta$; and (iii) slow down as agents' beliefs become more precise. Intuitively, when agents are confident about the state, they are confident that their peers have already coordinated on an action, so the indifference condition does not then not vary much. 

\textbf{Dynamics of aggregate play.} We now go from threshold beliefs to aggregate play. By the continuum law of large numbers,\footnote{See, e.g., \cite{sun2009individual}.} the time-$t$ distribution of posterior means in the population of agents is $\mu_{it} \sim N(\theta,\eta_t^2)$. Hence, aggregate time-$t$ play is simply given by
\[
\lambda_t(\theta)=\mathbb{P}_{\mu_{it} \sim N(\theta,\eta_t^2)}\Big(\mu_{it}\geq \mu_t^* \Big)=\Phi\left(\frac{\theta-\mu_t^*}{\eta_t}\right)
\]
Since we assumed limit learning, $\lim_{t \to \infty} \eta_t=0$. Taking $t\to\infty$ yields
\[
\lambda_\infty(\theta)=\mathbbm{1}[\theta \geq \mu_\infty^*]
\]
Hence, the limit belief threshold $\mu_\infty^*$ pins down limit aggregate play at state $\theta$. In particular, $\lambda_\infty(\theta)=RD(\theta)$ a.e. if and only if $\mu_\infty^*=c-(1/2)$. In words, the risky action is adopted by the whole population in the limit if and only if, in the complete information game in which the fundamental $\theta$ is known, playing the risky action is the best response to the {uniform conjecture} that $1/2$-proportion of players are playing the risky action. Hence, we call $c-(1/2)$ the \textit{risk dominance threshold}. 

We now walk through two simple but illustrative examples. 

\begin{example}[Global games]
\label{ex:brd}
%In this example only, we relax the limit learning assumption. 
For $\sigma>0$, let the learning process $\Sigma=(\sigma^2,\infty,\infty,\ldots)$ be such that agents only receive information at time $t=1$. Since $\eta_t^2=\sigma^2$ for all $t \in \mathcal{T}$, (\ref{main-text-mu-lom}) reduces to the stationary difference equation
\[
\mu_t^*=\mu_{t-1}^*+\sqrt{2}\sigma \Phi^{-1}(c-\mu_t^*)
\]
For any $\sigma>0$ and $\lambda_0 \in (0,1)$, this equation has a unique, globally attracting fixed point at $\mu_\infty^*(\sigma,\lambda_0)=c-(1/2)$. Hence, for any initial signal precision and initial aggregate play, the risk dominance \emph{threshold} is attained in the limit. However, without limit learning, limit play is not exactly risk-dominant since some agents remain mistakenly optimistic/pessimistic about the state. But, as the initial signal becomes arbitrarily precise,
\[
\lim_{\sigma \downarrow 0} \lambda_\infty(\theta|\sigma,\lambda_0)=RD(\theta).
\]
Indeed, in Appendix \ref{app:microf} we develop a two-player interpretation which recovers an equivalence between best-response dynamics and iterated deletion of interim strictly dominated strategies dynamics. Thus, the special case with learning only at $t=1$ and noisy signals thereafter is exactly the model of \cite*{carlsson1993global}. $\hfill \diamondsuit$
\end{example}
% Appendix \ref{app:microf} makes this connection precise. 

\begin{example}[Perfect learning]
\label{ex:comp_info}
Consider any learning process $\Sigma$ with $\sigma_1=0$, such that $\theta$ is common knowledge from time $t=1$ onwards. Suppose that $\theta+\lambda_0>c$. We show by induction that $\lambda_t=1$ for all $t\in \mathcal{T}$. For the base case, note that every agent $i$ strictly prefers $a_{it}=1$ if and only if $\theta+\lambda_0>c$, which holds by assumption. Hence, $\lambda_1=1$ is common knowledge. For the inductive step, assume $\lambda_{t-1}=1$. Then,
    \[
    \theta+\lambda_{t-1}-c=\theta+1-c\geq \theta+\lambda_0-c>0
    \]
    which implies $\lambda_t=1$. 
    The argument for $\theta+\lambda_0<c$ is analogous so $\lambda_t=0$ for all $t \in \mathcal{T}$. Hence, limit play depends on $\lambda_0$: 
$\lambda_\infty(\theta|\Sigma,\lambda_0)=\mathbbm{1}[\theta \geq c-\lambda_0].$
$\hfill \diamondsuit$
\end{example}

Examples \ref{ex:brd} and \ref{ex:comp_info} resemble the well-understood gap between small perturbations of common knowledge highlighted in the literature on static coordination games (Example \ref{ex:brd}). In our environment, complete information does not induce multiplicity because players' incentives are non-contemporaneous; but it induces \emph{history-dependence}: limit play is sensitive to initial play.\footnote{This is consistent with experimental evidence from coordination games which have largely focused on complete information settings; see, e.g. \cite*{van1991strategic}.}

What are we to make of these differing predictions? In what follows we offer a unified account of when risk-dominance or history-dependence obtains as a function of the learning process. It will turn out that Example \ref{ex:brd} is a special case of \emph{slow learning} since posterior precisions grow sub-quadratically (in Example \ref{ex:brd}, they do not change beyond $t = 1$).  On the other hand, Example \ref{ex:comp_info} is a special case of \emph{fast learning} since posterior precisions are super-quadratic (in Example \ref{ex:comp_info}, they ``jump immediately to infinity'' at time $t = 1$).\footnote{Translated into a static incomplete information game, this correspond to interspersing interim Bayesian updating with rounds of interim deletion of strictly dominated strategies. Appendix \ref{app:microf} makes this connection precise.}

\section{Characterization of limit play}
\begin{comment}
We begin with a useful comparative statics result that plays a key role in the subsequent analysis.

\begin{proposition}
\label{prop: monotone}
If $\Sigma \geq \Sigma'$ and $|\lambda_0-(1/2)|\leq |\lambda_0'-(1/2)|$, then 
\[
\Big|\mu_\infty^*(\Sigma,\lambda_0)-(c-1/2)\Big|\leq \Big|\mu_\infty^*(\Sigma',\lambda_0')-(c-1/2)\Big|
\]
\end{proposition}

In words, Proposition \ref{prop: monotone} states that either slower learning or more uncertain initial play yields a limit belief threshold that is closer to the risk dominance threshold. The key intuition underlying Proposition \ref{prop: monotone} stems from the previous observations that $(\mu_t^*)_t$ is monotone in $t$ and evolves \textit{slower} when learning is \textit{faster}. 
\end{comment}

Our main result gives a tight characterization of limit outcomes in terms of the learning rate. All proofs are in Appendix \ref{appendix:proofs}. To state our main result, we will use standard asymptotic notation: $f(t) = O(g(t))$ if there exist constants $\overline C < +\infty$ and $T < +\infty$ such that $t \geq T$ implies $f(t) \leq \overline C g(t)$. $f(t) = \Omega(g(t))$ if there exist constants $\underline C < +\infty$ and $T < +\infty$ such that $t \geq T$ implies $f(t) \geq \underline C g(t)$.

\begin{theorem} \label{thrm:char} The following gives a tight characterization of risk-dominance: 
\begin{itemize}
    \item[(i)] For all $\epsilon>0$ and all $\lambda_0 \in (0,1)$,
    \[
    \eta^{-2}_t = O(t^{2 - \epsilon}) \implies \lambda_\infty(\theta)=RD(\theta) \quad \text{a.e.} 
    \]
    \item[(ii)] For all $\epsilon > 0$ and all $\lambda_0 \in (0,1)$, 
    \begin{align*}
        \eta^{-2}_t = \Omega(t^{2 + \epsilon}) \implies &\lambda_\infty(\theta)=NRD(\theta) \\ &\text{for some positive measure interval $(\underline \theta, \overline \theta)$.}
    \end{align*}
\end{itemize}
\end{theorem}

Theorem \ref{thrm:char} identifies the critical threshold for limit play and, in so doing, develops a tight connection between sub-quadratic growth of posterior precisions and risk-dominance. Economically, it suggests that in the presence of (a little) inertia, whether and how shocks propagate to shape coordination outcomes depends sharply on learning speeds. We discuss the implications of each part in turn. 

Theorem \ref{thrm:char} (i) generalizes the global games prediction to the case where learning is sufficiently slow. Thus, our model with an endogenously changing state offers an alternate foundation for equilibrium selection. Moreover, it clarifies that sub-quadratic growth of posterior precision under private learning is both necessary and sufficient for risk-dominant selection. Thus, although the `vanishing noise limit' in static global games \citep{carlsson1993global,morris_shin_2003} is often interpreted as modeling players with vanishing uncertainty, this actually corresponds to slow learning in our dynamic model with endogenously changing states.\footnote{We are grateful to Muhamet Yildiz for pointing out this interpretation.} 

\begin{figure}[h!]
    \centering
      \caption{Path of aggregate play under different learning regimes} 
      \begin{quote}
    \vspace{-1em}
    \centering 
    \footnotesize Parameters: $\theta = 0.4$, $\lambda_0 = 0.75$, $c = 1$
    \end{quote} 
      \vspace{-1.0em}
    \subfloat[Sub-quadratic posterior precision]{\includegraphics[width=0.45\textwidth]{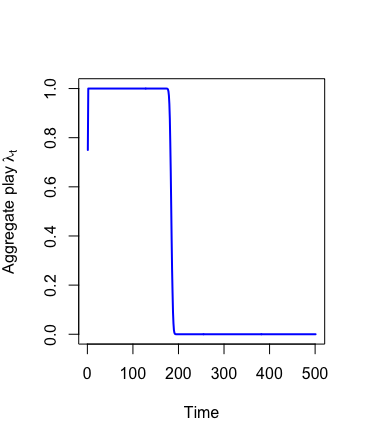}} 
    \subfloat[Super-quadratic posterior precision]{\includegraphics[width=0.45\textwidth]{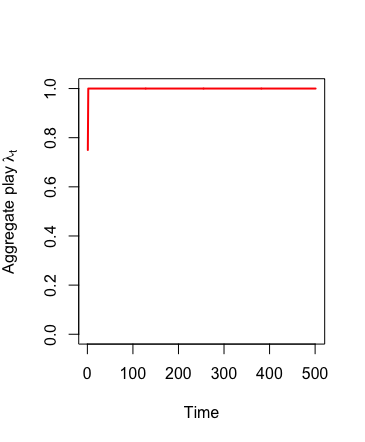}} 
    \label{fig:intro}
\end{figure}

Economically, this also implies that shocks to aggregate incentives (which we can interpret $\lambda_0$ as) do not persist when learning is sufficiently slow. 
Figure \ref{fig:intro} (a) illustrates a possible path of aggregate play under slow learning.\footnote{Notice that $(\lambda_t)_t$ `jumps' toward the risk-dominant profile. In Section \ref{sec:transition} we identify conditions under which the transition to the risk-dominant action is `gradual' versus `sudden'.} The true state is $\theta = 0.4$ and this governs the location of (heterogeneous) beliefs in the population. The cost of playing the risky action is $c = 1$ so under the conjecture that $1/2$-proportion of the population plays the risky action and common-knowledge of the state $\theta = 0.4$, the risk-dominant profile is to not play the risky action. Thus, following the initial shock $\lambda_0 = 0.75$, aggregate play converges back to the safe action.

Theorem \ref{thrm:char} (ii) breaks risk-dominance by showing that sufficiently fast learning can reintroduce the influence of initial play. This generalizes Example \ref{ex:comp_info}'s prediction from instantaneous learning to fast learning. Economically, this implies shocks to aggregate incentives are persistent when learning is fast. Figure \ref{fig:intro} (b) illustrates a typical path of aggregate play under super-quadratic learning. Under the conjecture that $\lambda_0 = 0.75$ proportion of players play the risky action and common-knowledge of the state $\theta = 0.4$, playing the risky action is indeed a best-response. 

Putting parts (i) and (ii) of Theorem \ref{thrm:char} together offers simple but novel economic implications for how the ex-ante joint distribution over shocks $(\lambda_0)$ and fundamentals ($\theta$) interact with learning speeds to shape limit outcomes. Figure \ref{fig:joint} depicts the regions of realizations of shocks and fundamentals under which shocks propagate so the limit play is the risky action (red region), or fizzle out so limit play is the safe action (blue region). Panel (a) depicts this for the case of sub-quadratic learning. As Theorem \ref{thrm:char} (i) establishes, the risk-dominant boundary is selected so limit play is history independent, and does not depend on the shock's realization. 
Panel (b) depicts the case of super-quadratic learning which \emph{rotates} the boundary clockwise around the point $(1/2, c-1/2)$.\footnote{We have depicted this so that the state selects initial-play dominance i.e., $\lambda_{\infty} = 1 \iff \theta \geq c - \lambda_0$. We will later establish (Proposition \ref{prop:implementable} that initial-play dominance and risk-dominance span possible limit outcomes implementable via any learning speed $(\sigma_t)_t$. More generally, for a given super-quadratic learning rate, this will lead to a (potentially non-linear) boundary between these two extremes.} Hence, when learning is fast, the regime might now survive when the state is weak as long as the initial shock is small (region $A$); on the other hand, the regime might now also fail when the state is strong when the initial shock is large (region $B$).

\begin{figure}[h!]
    \centering
    \caption{How learning speeds shape limit play}
      %\caption{Path of aggregate play under different learning regimes} 
      %\begin{quote}
    %\vspace{-1em}
    %\end{quote} 
      \vspace{-0.8em}
    \subfloat[Sub-quadratic]{\includegraphics[width=0.45\textwidth]{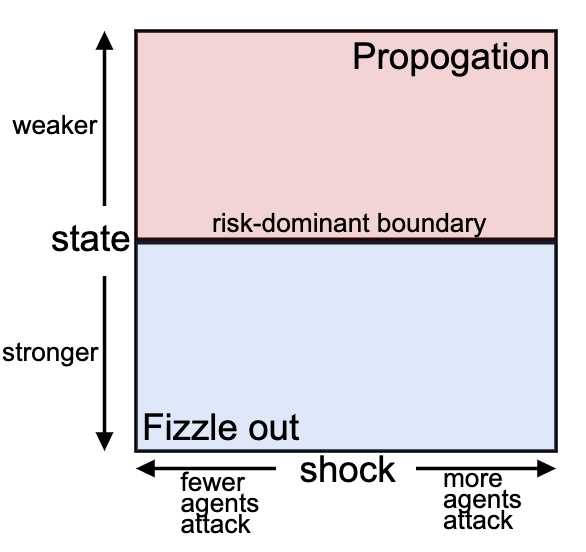}} 
    \subfloat[Super-quadratic]{\includegraphics[width=0.45\textwidth]{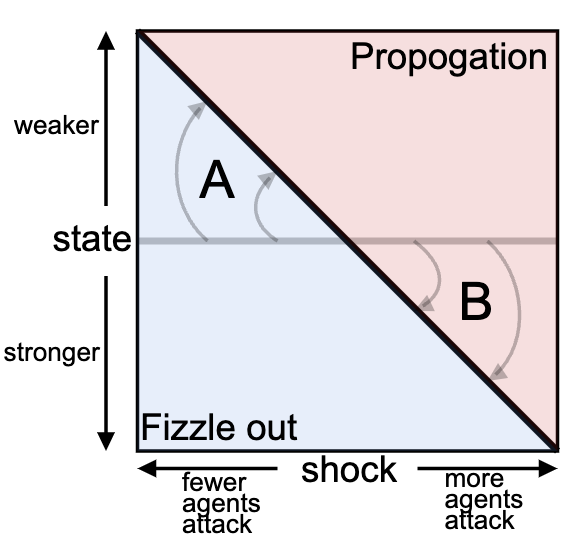}} 
    \label{fig:joint}
\end{figure}

\begin{comment}
Thus, whether fast or slow learning is preferred by the regime is ultimately dependent on whether the ex-ante distribution over fundamentals and shocks is puts more probability on region $A$ or region $B$.

Consider an ex-ante joint distribuiton 

Fast (super-quadratic) learning is better for the regime if and only if A$>$B.  \textcolor{blue}{Kei: Fix any joint distribution $F$ that is symmetric around $(\frac{1}{2}, c-\frac{1}{2})$, i.e., the pdf satisfies $f(\lambda_0, \theta) = f(1 - \lambda_0, 2c - 1 - \theta)$. Under $F$, we have $A = B$ because $\mu_\infty^*(1 - \lambda_0) = 2c - 1 - \mu_\infty^*(\lambda_0)$. Let $F_{|\theta}$ denote the conditional distribution of $\lambda_0$ given $\theta$. For every $\theta$, take some $\tilde F_{|\theta}$ that is dominated by $F_{|\theta}$ in the sense of FOSD. Define the joint distribution $\tilde F$ such that $\tilde f(\lambda_0, \theta) = \tilde f_{|\theta}(\lambda_0) f_{\theta} (\theta)$. Under $\tilde F$, we have $\tilde A > A = B > \tilde B$.
}
\end{comment}

That sufficiently fast learning can break risk-dominance is related to the idea that public signals can yield non-risk dominant play even as uncertainty about the fundamentals vanish; see \cite{morris_shin_2003} Chapter 3.2. In both cases, fast private learning or public signals enable players to quickly approach common knowledge of the fundamental and their peers' play, thus enabling coordination on the non-risk dominant action. Our result complements and enhances this insight by providing an explicit dynamic foundation, as well as a precise characterization of how fast learning needs to be for this intuition to hold. This allows us to classify several canonical learning environments and their resulting limit outcomes.

%Moreover, in this super-quadratic regime, the {order} at which signals of different precisions arrive over time will generally matter for the limit profile of play. For instance, initial periods of noisy signals followed by subsequent periods of super-quadratic signals selects the payoff-dominant action on a wider set of states than if those noisy signals arrived at later periods. 

\textbf{Intuition for Theorem \ref{thrm:char}.} Recall our previous observation that the speed of evolution of the time-$t$ indifferent belief $(\mu_t^*)_t$ is inversely related to the speed of learning: when learning is slow enough, $(\mu_t^*)_t$ gradually converges to the risk dominant threshold $\mu_\infty^*=c-(1/2)$ from any initial condition. However, when learning is fast enough, the dynamics of $(\mu_t^*)_t$ freeze, and consequently $\mu_\infty^*$ is bounded away from the risk dominant threshold. The proof of Theorem \ref{thrm:char} formalizes this intuition by explicitly constructing fictitious threshold sequences that converge to (non)-risk dominant play and matching them to the appropriate learning rates.  

Figure \ref{fig:slow_learning} illustrates the forces underlying Theorem \ref{thrm:char} (i) for the case where $\lambda_0<1/2$. Each panel corresponds to a different time. In each panel, the dotted line represents the state $\theta$, and the red curve is the distribution of posterior means $\mu_t$. The blue line is the threshold belief $\mu_t^*$, and the shaded area is the mass of agents $\lambda_t$ who take the risky action because their beliefs exceed the threshold. As time passes, agents learn and $\mu_t$ contracts around $\theta$. In Figure \ref{fig:slow_learning}, learning and hence the speed of this contraction is slow. Consequently, the evolution of $\mu_t^*$ is relatively fast. As $t \to \infty$, $\mu_t^*$ converges to $c-(1/2)$ and $\mu_t$ converges to a point mass on $\theta$ due to limit learning. Since $\theta \geq c-(1/2)$, aggregate play converges to the risk dominant (risky) action.

\begin{figure}[H]  
\centering
\captionsetup{width=0.9\linewidth}
\vspace{-0.8em}
    \caption{Slow learning implies risk dominance} \includegraphics[width=0.9\textwidth]{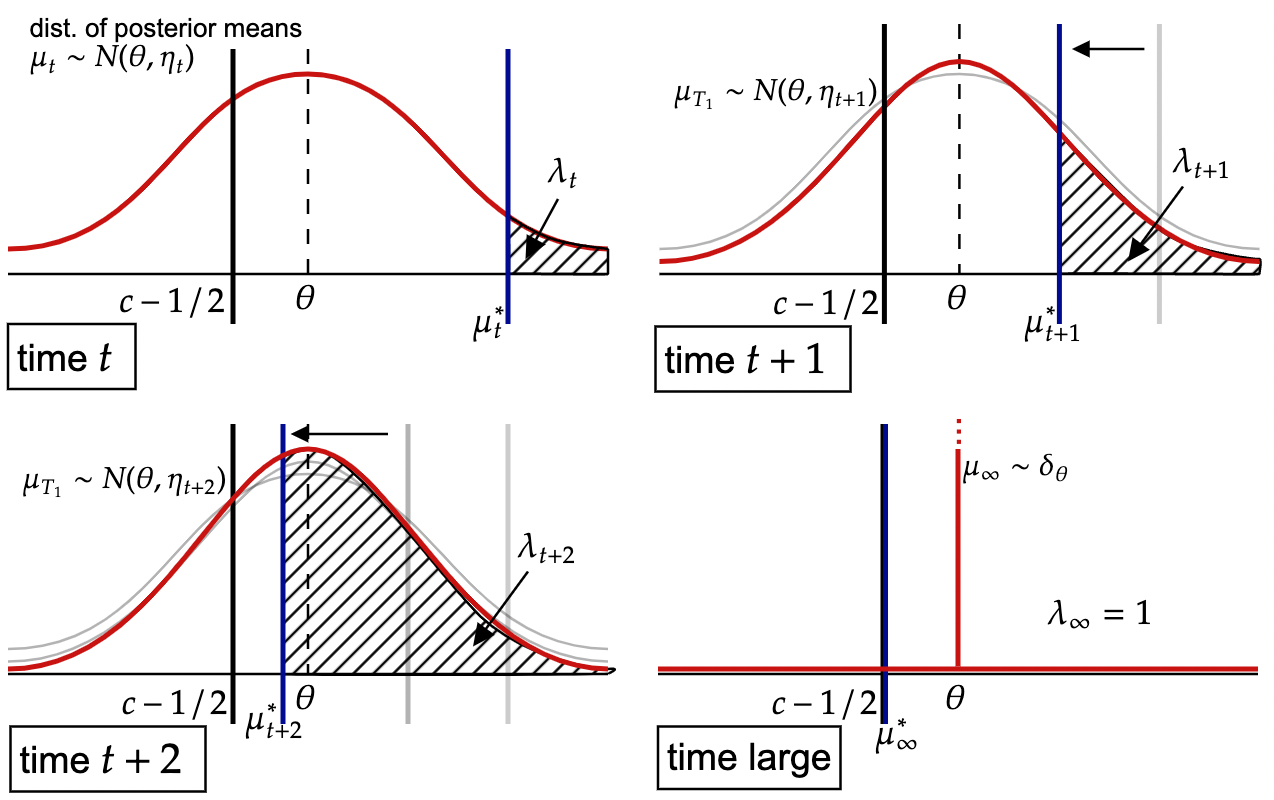}
     \label{fig:slow_learning}
\end{figure}

Figure \ref{fig:fast_learning} illustrates the forces underlying Theorem \ref{thrm:char} (ii) for the case where $\lambda_0<1/2$. With super-quadratic learning, the speed of contraction of $\mu_t$ is fast. Consequently, the evolution of $\mu_t^*$ is relatively slow. As $t \to \infty$, $\mu_t^*$ converges to a limit threshold $\mu_\infty^*$ that is bounded away from $c-(1/2)$, and since $\theta \in (c-(1/2),\mu_\infty^*)$, aggregate play converges to the non-risk dominant (safe) action.

\begin{figure}[h!]  
\centering
\captionsetup{width=0.9\linewidth}
    \caption{Fast learning implies non-risk dominance} \includegraphics[width=0.9\textwidth]{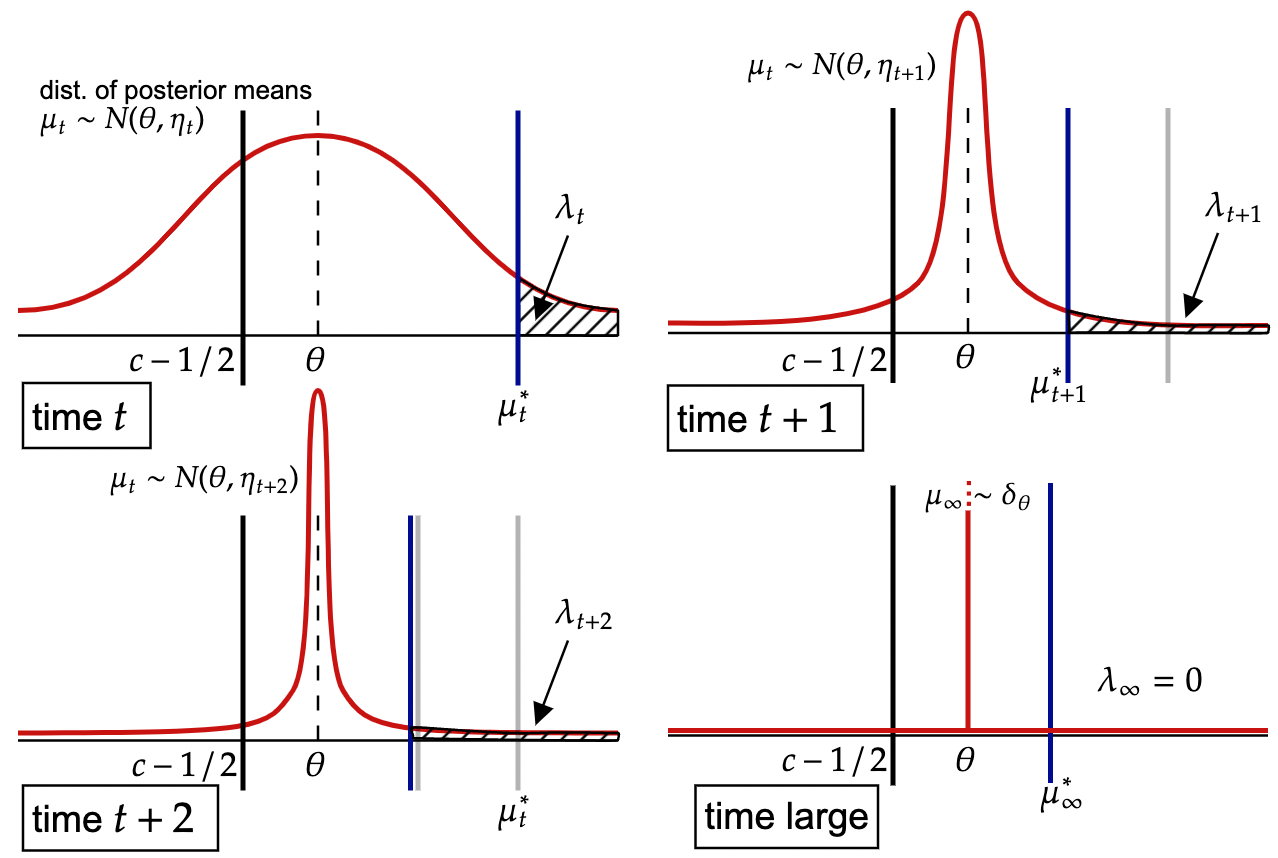}
     \label{fig:fast_learning}
\end{figure}

Several corollaries follow from Theorem \ref{thrm:char}.

\begin{corollary}[Risk-dominance is a tail event] \label{cor:tail}
    Consider any sub-quadratic learning process $\Sigma = (\sigma_t^2)_{t \in \mathcal{T}}$. Then, for any finite $T < +\infty$ and any $(\sigma'_t)_{t \leq T}$, the risk-dominant action profile is selected almost surely in the limit under the learning process $\Sigma' := [((\sigma'_t)^2)_{t \leq T}, (\sigma_t^2)_{t > T}]$. 
\end{corollary}

This gives a dynamic perspective on the gap between almost-common knowledge and common knowledge of payoffs in static incomplete information games. 
% Indeed, Corollary \ref{cor:tail} suggests that the Gaussian form of private signals is {not} necessary since, for (i) an arbitrary noise structure (with bounded variance); and (ii) a sufficiently regular common prior, as $T$ grows large the distribution over posterior beliefs after $T$ grows close to the Gaussian distribution by an appropriate application of the Bernstein-von Mises Theorem \citep{van2000asymptotic}. 

\begin{corollary}
The following are limit outcomes of canonical learning processes:
\label{cor:prop_1}
\begin{itemize}
    \item[(i)] (IID Gaussian signals) Suppose $\Sigma=(\sigma^2,\sigma^2,\ldots)$. Then for any $\lambda_0 \in (0,1)$ and $\sigma>0$, limit aggregate play is the risk dominant action.
    % \item[(ii)] (Static global games) Suppose $(\sigma_t)_t$ is such that $\sigma_1 = \sigma$ and $\sigma_t = +\infty$ for all $t \geq 2$. Then for any fixed $\sigma>0$, $\mu_\infty^*(\sigma)=c-(1/2)$, and hence the risk dominance equilibrium is selected as $\sigma \downarrow 0$.
    % Although we have relaxed Assumption \ref{assumption: lim_learning}, by an analogous argument as in Proposition \ref{thrm:char}, \footnote{Formally, for any fixed $\sigma>0$, note that $\eta_t^2=\sigma^2$ is constant in $t$, and the threshold update mapping $\mu(x)$ defined implicitly by $\mu(x)+\Phi((\mu(x)-x)/\sqrt{2}\sigma)=c$ is a contraction mapping with unique global attactor at $c-(1/2)$. Hence, $\lim_{\sigma \downarrow 0} \mu_\infty^*(\sigma)=\lim_{\sigma \downarrow 0} (c-1/2)=c-1/2$. Finally, for any fixed $\theta$, $\lim_{\sigma \downarrow 0}\lambda_\infty(\theta,\sigma)=\lim_{\sigma \downarrow 0}\Phi((\theta-(c-1/2))/\sigma)=\mathbbm{1}(\theta \geq c-1/2)$ a.s., as desired.}
    % \[
    % \lim_{\sigma \downarrow 0} \mu_\infty*(\sigma)=c-(1/2)
    % \]    
    \item[(ii)] (Social learning) At each time $t \geq 2$, player $i$ exchanges information with a randomly matched player $j$. That is, in the first period, each play receives a private signal $\sigma^2_1 = \sigma^2$ and in all subsequent periods, players' posterior precisions double: $\sigma^2_t = \eta^2_{t-1}$. Then the non-risk dominant action is played on a positive measure of states. In fact, we will later show that if $\sigma$ is small, limit play converges to \emph{initial play dominance} where $\lambda_{\infty}(\theta) = 1$ if and only if $\theta \geq \lambda_0$.
    %Thus, defining  $\Sigma(\sigma)$ to be: $\sigma_1^2:=\sigma^2$ and for all $t\geq 2$, $\sigma_t^2:=\eta_{t-1}^2$, we have that for any $\lambda_0$, there exists some $\underline{\sigma}>0$ such that under $\Sigma(\sigma)$ for any $\sigma < \underline\sigma$, there exists an open interval of states at which limit aggregate play is the non-risk dominant action.
\end{itemize}
\end{corollary}

Corollary \ref{cor:prop_1} (i) shows that limit risk dominant play independent of initial conditions holds for a canonical informational environment. Corollary \ref{cor:prop_1} (ii) studies a social learning setting where at time $1$, each agent observes an iid private signal, and at each subsequent time, each agent is matched at random with another player in the population and observes the other player's information. This provides a simple informational setting where, if the initial signal is informative enough, non risk-dominant coordination arises. 

\textbf{Comparative statics.} Theorem \ref{thrm:char} identified a sub- and supercritical regime for how learning speeds impact limit play. We now develop comparative statics on how the set of equilibria vary with primitives. Recall that limit play is jointly pinned down by the limit threshold $\mu^*_{\infty}$ and the fundamental $\theta$: $\lambda_{\infty} = \mathbbm{1}(\theta \geq \mu^*_{\infty})$. Thus, $\mu^*_{\infty} = c-1/2$ corresponds to risk-dominance; as $|\mu^*_{\infty} - (c - 1/2)|$ becomes large, the non-risk dominant limit profile is selected on a wider set of states.

\begin{proposition}[Comparative statics]\label{prop: comp stat}  
The following are comparative statics on $|\mu^*_{\infty} - (c - 1/2)|$: 
    \begin{itemize}
        \item[(i)] \textbf{Extreme initial play induces extreme limit play.} For any $|\lambda_0-(1/2)|\leq |\lambda_0'-(1/2)|$ and any $\Sigma$, the limit threshold is further from the risk-dominance threshold when initial play is more extreme: 
        \[
        \Big|\mu_\infty^*(\Sigma,\lambda_0)-(c-1/2)\Big|\leq \Big|\mu_\infty^*(\Sigma,\lambda_0')-(c-1/2)\Big|.
        \]

        \item[(ii)] \textbf{Monotonicity in learning speeds.} 
        For any $\Sigma \geq \Sigma'$, the limit threshold is further from the risk-dominance threshold when learning is quicker: 
        \[
        \Big|\mu_\infty^*(\Sigma,\lambda_0)-(c-1/2)\Big|\leq \Big|\mu_\infty^*(\Sigma',\lambda_0)-(c-1/2)\Big|.
        \]
        
        \item[(iii)] \textbf{Early learning rates matter more.}
        Fix any $\Sigma = (\sigma_t^2)_t$ and define the pairwise $s\leftrightarrow s'$ permutation $\Sigma^{s\leftrightarrow s'} := (\tilde\sigma_t^2)_t$ where
        \[
        \tilde\sigma_t^2 = \begin{cases}
            \sigma_t^2 &\text{if $t \neq s, s'$} \\
            \sigma_{s'}^2 &\text{if $t = s$} \\
            \sigma_s^2 &\text{if $t = s'$}. 
        \end{cases}
        \]
        Then,  the limit threshold is further from the risk-dominance threshold when informative signals are front-loaded: for any $s \geq s'$ and $\sigma_s^2 \geq \sigma_{s'}^2$,
        \[
        \Big|\mu_\infty^*(\Sigma,\lambda_0)-(c-1/2)\Big|\leq \Big|\mu_\infty^*(\Sigma^{s\leftrightarrow s'},\lambda_0)-(c-1/2)\Big|.
        \]
    \end{itemize}
\end{proposition}

Proposition \ref{prop: comp stat} is intuitive. Part (i) states that the non-risk dominant play expands (i.e., on a wider set of states the non-risk dominant equilibrium is played in the limit) whenever initial play is more extreme on either side of $1/2$; Part (ii) states that non-risk dominant play expands when signals are more precise at each time step (although, per Theorem \ref{thrm:char}, this only has bite in the superquadratic regime); Part (iii) shows that non-risk dominant play expands when the learning process is permuted so that more informative signals arrive early, and less informative signals arrive late. Intuitively, learning and aggregate behavior at early time steps are more important because they propagate into the future.

\textbf{Possible limit play.} Theorem~\ref{thrm:char} states that when learning is sufficiently fast, the limit action is non-risk dominant at some open interval of states. Proposition \ref{prop: comp stat} established how this interval varies as primitives change. This raises the following question: what kinds of limit play can be rationalized by some learning rate? 

Say that a profile is $\lambda_0$-dominant if, under common-knowledge of $\theta$, players play as-if they best-respond to the conjecture that $\lambda_0$ proportion of players play the risky action.\footnote{This is equivalent to the definition of $p$-dominance in \cite*{morris1995p}.} Our next result establishes that any profile between $\lambda_0$-dominance and risk-dominance can be rationalized by some learning speed.

\begin{proposition}[The image of learning rates]
\label{prop:implementable} 
For any $\lambda_0 \in (0,1)$,
\[
\bigcup_{\Sigma} \lambda_\infty(\theta|\Sigma,\lambda_0) = \bigcup_{\substack{\text{$\mu^*$ between} \\ \text{$c - \lambda_0$ and $c-1/2$}}} \mathbbm{1}(\theta \geq \mu^*).  
\]
In particular, for any threshold $\mu^*$ between $c-\lambda_0$ and $c-(1/2)$, there exists some learning process $\Sigma\geq 0$ such that:
\[
\lambda_\infty(\theta|\Sigma,\lambda_0)=\mathbbm{1}(\theta \geq \mu^*) \quad \text{a.e.}
\]
\end{proposition}

The intuition behind Proposition \ref{prop:implementable} is that for any time $t$, we can arbitrarily restrict the evolution of the system by choosing a sufficiently precise signal for time $t$. That is, any level of coordination on the non-risk dominant action can be rationalized by some sequence of signal noises.\footnote{Recall that $\mu_1^*=c-\lambda_0$ and $\mu_t^*$ evolves monotonically away.} Combined with our comparative static on learning speeds, this also implies $\lim_{\Sigma \downarrow 0} \lambda_{\infty}(\theta|\Sigma,\lambda_0) = \mathbbm{1}(\theta \geq c - \lambda_0)$---that is, in the vanishing noise limit along the whole path, limit play converges to that of the complete environment with all players best-responding to $\lambda_0$. 

Proposition \ref{prop:implementable} also highlights the following asymmetry between initial play and learning rates: when initial play is large ($\lambda_0 > 1/2$ large), slow learning hurts efficiency since the set of fundamental states at which the payoff-dominant action is played shrinks to $\{\theta \geq c - 1/2\}$; conversely, when initial play is small ($\lambda_0 < 1/2$ small), \emph{slow learning improves efficiency} since on the set of fundamentals $\{\theta \geq c - 1/2\}$ the payoff-dominant action is played in the limit.

\section{Transition dynamics}\label{sec:transition}
We have thus far focused on limit play. We now turn to studying the connection between the speed of learning and the time path of aggregate play. We will consider the iid signal case such that the learning process $\Sigma(\sigma)$ is parametrized by the standard deviation of each period's signal $\sigma>0$. We focus on the iid case because it is a canonical model of learning, and makes results easier to state since $\Sigma$ is now parametrized by a scalar; nonetheless, we expect that the qualitative features of our results extend beyond this case. Within this iid learning environment, recall that Corollary~\ref{cor:prop_1} (i) implies that the limit action is risk dominant for all fundamentals and initial conditions. The following result shows that the path of aggregate play exhibits qualitatively different behavior depending on learning rates. 

% The previous section developed the connection between the speed of learning and the action profile played in the limit. We now focus on the regime in which the risk dominant equilibrium is selected (i.e., when learning is sufficiently slow) and show how the speed of learning matters for the transition path of $(\lambda_t)_t$. To fix ideas, consider the iid signal case in which $\sigma_t = \sigma > 0$ for all $t$ such that $\eta_t^2 = O(1/t)$. We show that    

\begin{proposition}\label{prop: dynamics} 
Fix any $\lambda_0 \in (0,1)$ and $\theta$ between $c-1/2$ and $c-\lambda_0$. Define $\overline{\beta} :=1+|c-(1/2)-\theta|/2$. For any $\epsilon > 0$, $\alpha \in (0, 1)$, and $\beta \in (1, \overline \beta)$, there exists $\underline \sigma, \overline \sigma > 0$ such that:
    \begin{itemize}
        \item[(i)] \textbf{Sudden transition.} For all $\sigma < \underline \sigma$, there exists $T(\sigma)\geq 1$ such that: for all $1 \leq t \leq \alpha T(\sigma)$, $|\lambda_t(\theta)-NRD(\theta)|<\epsilon$ and for all $t \geq \beta T(\sigma)$, $|\lambda_t(\theta)-RD(\theta)| < \epsilon$; 
        \item[(ii)] \textbf{Gradual transition.} For all $\sigma > \overline \sigma$ and all $t \geq 1$, $|\lambda_{t+1}(\theta) -\lambda_t(\theta)| < \epsilon$.
    \end{itemize}
\end{proposition}

Proposition \ref{prop: dynamics} says that, within the regime of iid learning, sufficiently fast learning yields a \emph{sudden} transition from aggregate non-risk dominant to risk dominant play that occurs around some finite time. Conversely, sufficiently slow learning yields sufficiently \emph{gradual} transitions in aggregate play.

\begin{figure}[h!]  
\centering
\captionsetup{width=0.9\linewidth}
    \caption{Mechanism underlying transition regimes} \includegraphics[width=0.9\textwidth]{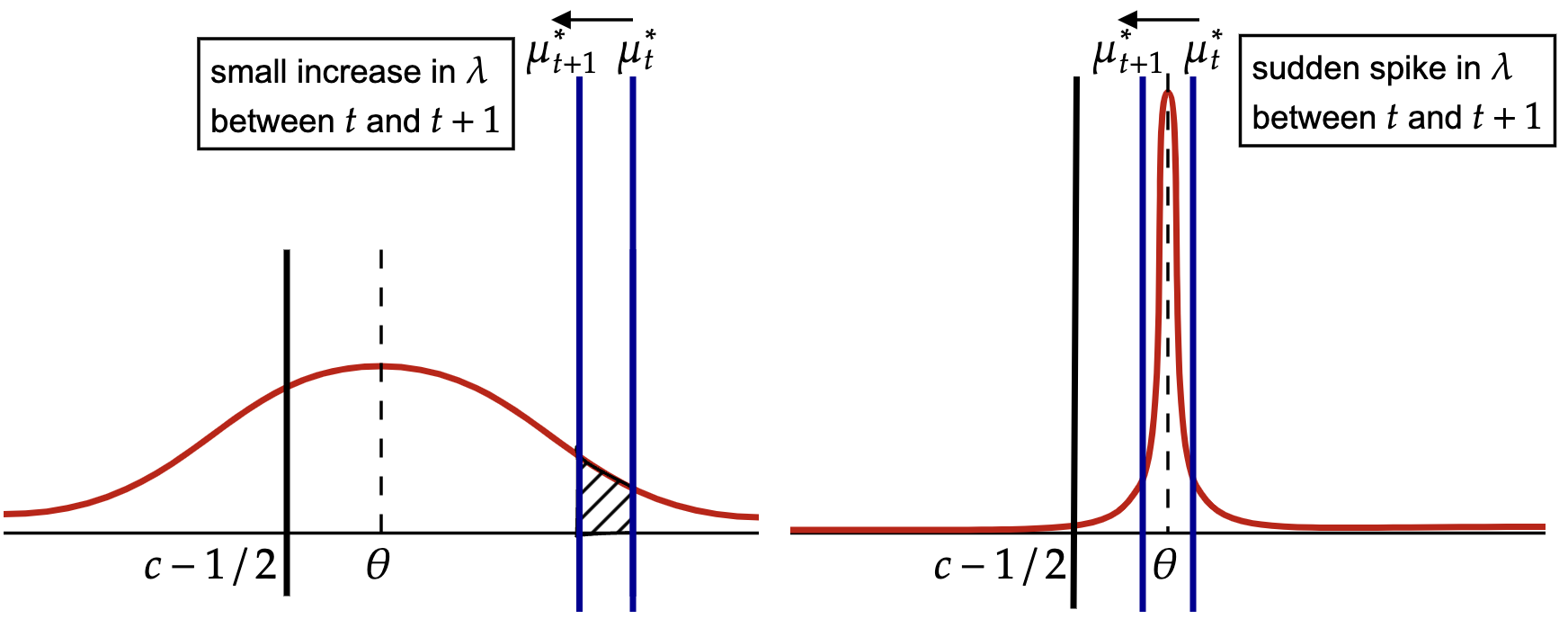}
     \label{fig:dynamic_intuition}
\end{figure}

The intuition behind Proposition \ref{prop: dynamics} is illustrated in Figure \ref{fig:dynamic_intuition}. In the case where learning is slow, the distribution of posterior means is dispersed so as the threshold changes, a relatively small mass of new agents take the risk dominant action. Conversely, when learning is fast, the distribution of posterior means contracts around $\theta$ so that there is a sudden spike in agents taking the risk dominant action when it is crossed.

Figure \ref{fig:dynamics} numerically illustrates the time path of play $(\lambda_t)_t$ for iid learning. Panel (a) shows the dynamics when $\sigma = 1$; panel (b) shows the dynamics when $\sigma = 0.01$. In particular, when learning is fast, an analyst would observe next-to no change in the measure of agents playing the risky action before a sudden spike. 

\begin{figure}[h!]
    \centering
      \caption{Gradual vs sudden transition to risk dominant play} 
      \begin{quote}
    \vspace{-1em}
    \centering 
    \footnotesize Parameters: $\theta = 0.6$, $\lambda_0 = 0.2$, $c = 1$.
    \end{quote} 
      \vspace{-0.8em}
    \subfloat[$\sigma = 1$: gradual transition]{\includegraphics[width=0.45\textwidth]{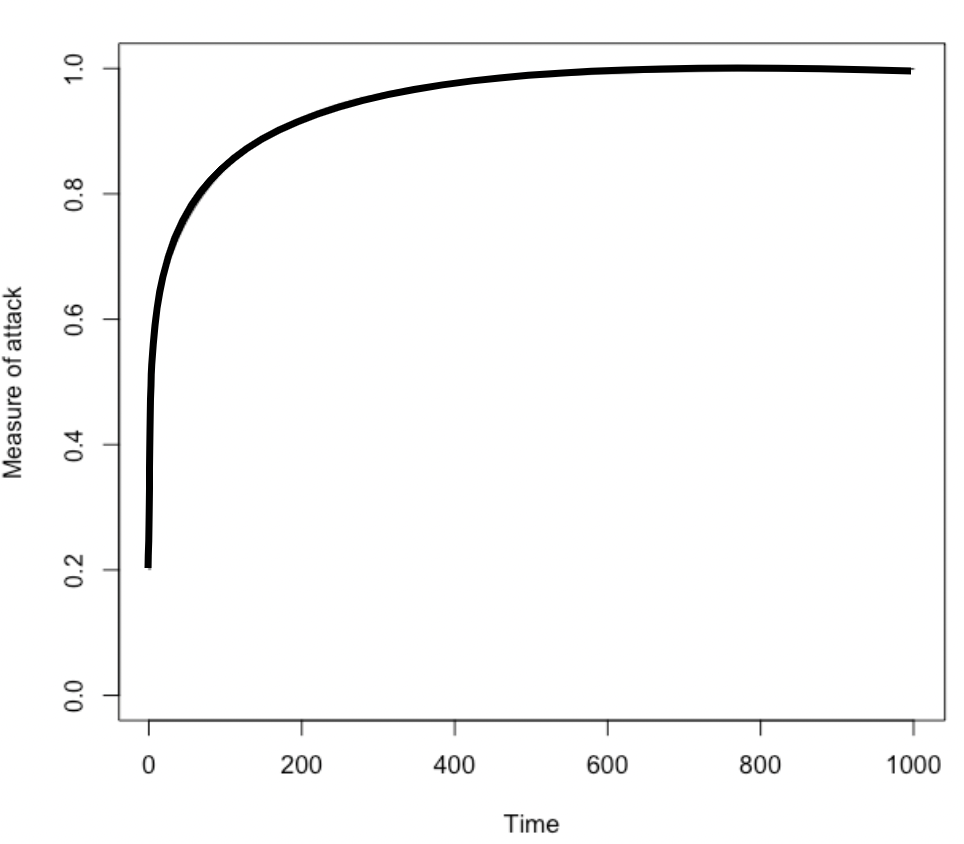}} 
    \subfloat[$\sigma = 0.01$: sudden phase transition]{\includegraphics[width=0.45\textwidth]{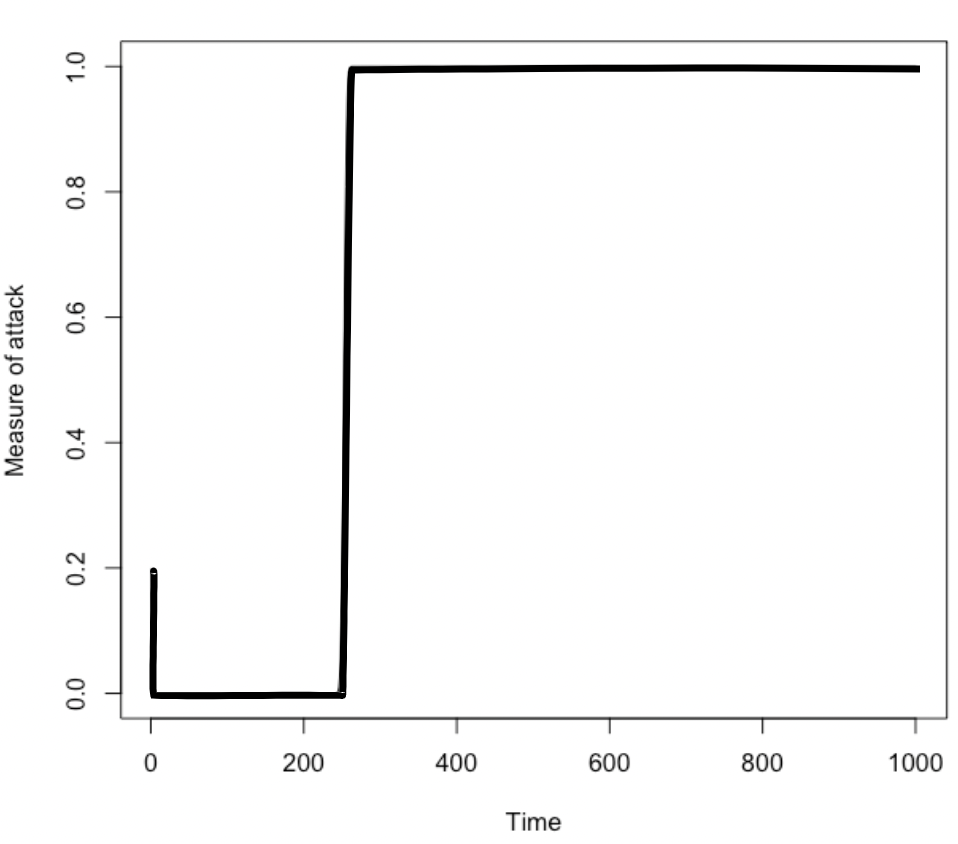}} 
    \label{fig:dynamics}
\end{figure}

Proposition \ref{prop: dynamics} has a subtle connection to the extant literature on coordination games. A common theme in this literature is that heterogeneity---in preferences, beliefs, etc.---across players can induce uniqueness in static supermodular games \citep{morris2006heterogeneity}. A dynamic version of this insight is that with lots of hetrogeneity, as aggregate incentives shift, the set of equilibria vary smoothly---this is analogous to Part (i) of Proposition \ref{prop: dynamics}; by contrast, with little heterogeneity, the set of equilbiria can vary discontinuously---this is analogous to Part (ii) of Proposition \ref{prop: dynamics}. The key insight of Proposition \ref{prop: dynamics} is that learning rates induce time-varying hetrogeneity in beliefs which, in turn, matters for transition dynamics.\footnote{\cite{koh2022speed} relates heterogeneity to the speed of contagion in a networked population. This can be equivalently interpreted (see, e.g., \cite{morris1997interaction}) as best-response dynamics or interim deletion of strictly dominated strategies where each player's beliefs and higher-order beliefs are encoded by the network. However, \cite{koh2022speed} study speed of contagion for graphs/beliefs which are held fixed i.e., in the incomplete information game interpretation, there is no learning.}  

Economically, Proposition \ref{prop: dynamics} offers a simple alternate interpretation of sudden changes in the time-series of aggregate play: when learning is precise, players have very similar beliefs by the time the ``tipping point'' arrives and the state has increased to the point that is optimal for such players to play the risky action. Sudden spikes have been commonly interpreted as ``equilibrium shifts'' driven by shocks or new public information \citep{chwe2013rational,morris2019crises} interacting with uniform rank beliefs. Our result offers an alternate explanation driven by private learning and time-varying heterogeneity.

\section{Concluding Remarks}
We have developed a simple but novel model of inertial coordination games. This allowed us to combine and unify insights from the literature on best-response dynamics in complete-information coordination games \citep{crawford1995adaptive} and global games \citep{carlsson1993global}. 
We think that inertial coordination games serve as a natural description of many  `real world' coordination environments with both shrinking uncertainty about fundamentals (e.g., solvency of a bank, strength of a currency regime, quality of a networked product) as well as an endogenously changing state shaped by past play (e.g., past withdrawals, short positions, past sales). At the same time, we also showed that inertial coordination games share much of the game-theoretic logic underpinning coordination and learning in games.\footnote{We show in Appendix \ref{app:small lags} that for a fixed speed of learning (as a function of physical time), limit equilibrium predictions coincide as inertia vanishes (gap between time periods shrink). In this regard, inertial coordination games generalizes the setting with contemporaneous coordination motives. Nonetheless for any fixed level of inertia, the sub/super-quadratic boundary continues to hold.}
In this regard, we offer a different perspective on `history versus expectations'  \citep{krugman1991history} in economic modeling---in the presence of inertia, expectations \emph{about} history shape coordination outcomes. In particular, the speed at which players grow confident is a crucial determinant for whether coordination succeeds---the risk-dominant action is played in the limit if and only if posterior precisions grow sub-quadratically (Theorem \ref{thrm:char}). This delivered a suite of rich predictions of whether and to what extent initial shocks can propagate. In the super-quadratic regime, we showed that any profile of limit play between risk- and initial play-dominance can be rationalized by \emph{some} learning speed. In the sub-quadratic regime, we further identified qualitatively different behavior in the path of aggregate play (Proposition \ref{prop: dynamics}): when learning is slow, transition dynamics are gradual; when learning is fast, sudden.

% First, our results uncover an asymmetry in the effect of faster learning: when initial play is large above $1/2$ ($\lambda_0 > 1/2$) faster learning leads the risky action to be played in the limit at a strictly larger set of states. Conversely, when initial play is below $1/2$, fast learning leads the risky action to be played in the limit at a strictly smaller set of states.  

% Second, while our results pin down the relationship between signal noise and limit outcomes of attack, we would like to better understand how attack dynamics away from the limit change as a function of noise. Specifically, we have the following conjecture:

% \begin{conjecture}[Contagion dynamics along the path of play]
% There exist noise cutoffs $\underline{\sigma},\overline{\sigma}>0$ such that:
% \begin{itemize}
%     \item[(i)] \textbf{[Jump to limit action when noise is small]}: For all $\sigma<\underline{\sigma}$, there exists $T\geq 1$ and $\epsilon>0$ such that: for all $t<T$, $|\lambda_t-\lambda_0|<\epsilon$, and for all $t>T$, $|\lambda_t-\lambda_\infty|<\epsilon$.
%     \item[(ii)] \textbf{[Gradual convergence when noise is large]}: For all $\sigma>\overline{\sigma}$, there exists $T\geq 1$ such that for all $t>T$, $\lambda_t\leq O(1/t)$.
% \end{itemize}
    
% \end{conjecture}

% Second, we may investigate selection in learning by introducing correlation between attacking and receiving information. This may be related to self-confirming equilibrium selection.

\newpage
\appendix 

\begin{center}
    \large 
    \textbf{APPENDIX TO INERTIAL COORDINATION GAMES}
\end{center}

\section{Omitted proofs} \label{appendix:proofs}
Section~\ref{appendix: dynamics} contains several useful lemmas that characterize the dynamics of our model. Sections~\ref{appendix: thrm1}-\ref{appendix: a5} collect all remaining proofs. We sometimes assume $\lambda_0 > \frac{1}{2}$ without loss of generality because symmetric arguments apply to $\lambda_0 < \frac{1}{2}$. In particular, the proofs of Theorem~\ref{thrm:char} and Proposition~\ref{prop:implementable} consist of several lemmas, as described in Figure~\ref{fig:proof_map}.

\begin{figure}[h!]  
\centering
\captionsetup{width=1\linewidth}
    \caption{Roadmap for proofs of Theorem \ref{thrm:char} and Proposition \ref{prop:implementable}} \includegraphics[width=1\textwidth]{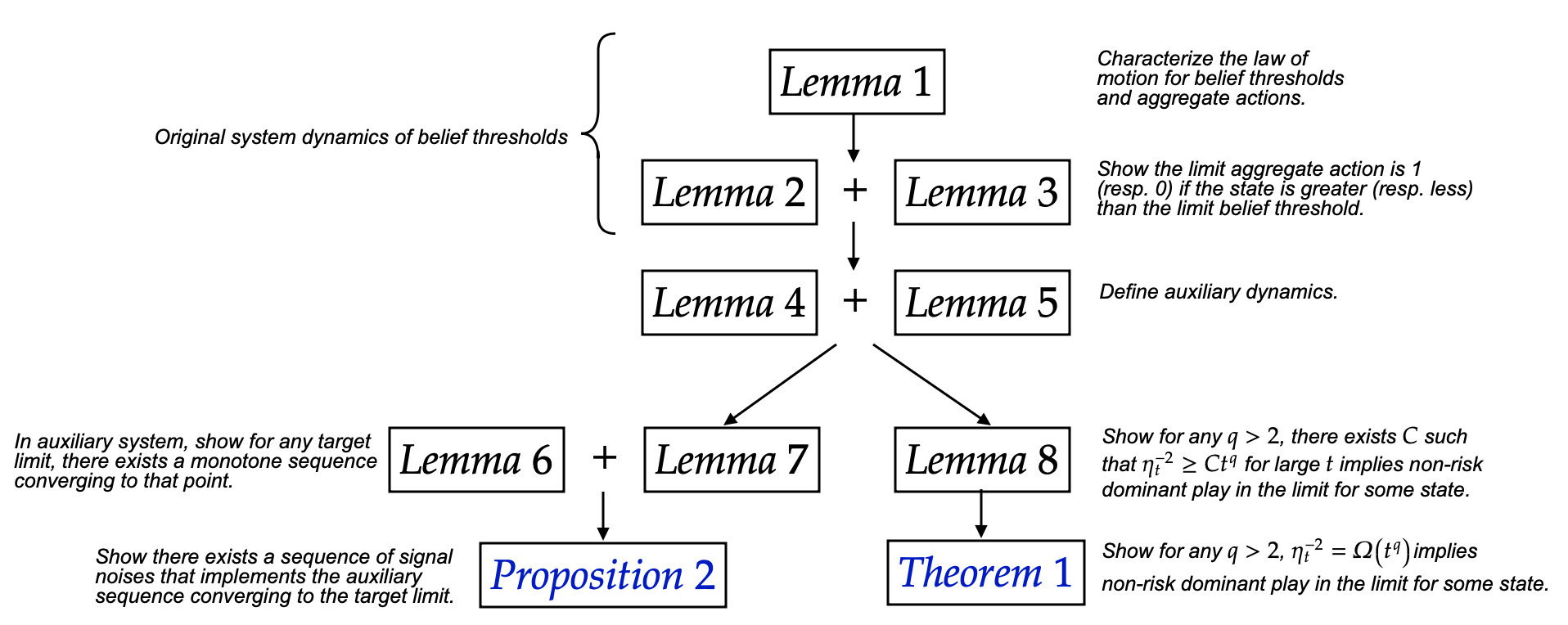}
    \label{fig:proof_map}
\end{figure}

\subsection{System Dynamics} \label{appendix: dynamics} 
Lemmas~\ref{lem: lom}-\ref{lemma: mu star limit} show that in equilibrium, agents use a cutoff strategy with respect to posterior means. Lemmas~\ref{lem:gamma_lom}-\ref{lem: A_gamma_1_properties} study useful auxiliary system dynamics.

\begin{lemma}
    \label{lem: lom}
    At each time $t\geq 1$ and for each player $i \in [0,1]$, $a_{it}=1$ if and only if $\mu_{it}\geq \mu_t^*$, where $\mu_1^*:=c - \lambda_0$ and for all $t>1$, $\mu_t^*$ is the unique solution to the nonlinear difference equation
    \begin{equation}
        \label{eq: mu star}
        \tag{$\mu^*$-LoM}
        \mu_t^*+\Phi\left(\frac{\mu_t^*-\mu_{t-1}^*}{\sqrt{\eta_{t-1}^2+\eta_t^2}}\right)=c
    \end{equation} 
    Hence, for each $t\geq 1$, 
    \begin{equation}
        \label{eq: lam LoM}
        \tag{$\lambda$-LoM}
        \lambda_t(\theta)=\Phi\left(\frac{\theta-\mu_t^*}{\eta_t}\right)
    \end{equation}
    %Agent $i$ takes $a_{it}=1$ if and only if $\mu_{it}\geq\mu_t^*$ holds, where $\mu_t^*$ satisfies
    %\begin{equation}
    %    \label{eq: mu star}
    %    \mu_t^*+\Phi\left(\frac{\mu_t^*-\mu_{t-1}^*}{\eta_{t-1}\sqrt{1+(\eta_t/\eta_{t-1})^2}}\right)=c,
    %\end{equation}
    %and $\mu_0^*=c-\lambda_0$. 
\end{lemma}

\begin{proof}[Proof of Lemma \ref{lem: lom}]
    We prove Lemma~\ref{lem: lom} by induction. At $t=1$, each agent $i$'s posterior belief is given by $\theta|\mathcal{F}_{i1}\sim N(x_{i1},\sigma_1^2)$. Hence, she takes $a_{i1}=1$ if and only if $\mathbb{E}_{i1}[\theta_1]=\mu_{i1}=x_{i1}\geq c - \lambda_0=\mu_1^*$. Since $x_{i1}|\theta\sim N(\theta,\sigma_1^2)$,
    \begin{align*}
        \lambda_1(\theta)=\mathbb{P}_\theta(x_{i1}\geq\mu_1^*)=1-\Phi\left(\frac{\mu_1^*-\theta}{\sigma_1}\right)=\Phi\left(\frac{\theta-\mu_1^*}{\sigma_1}\right).
    \end{align*}
    Hence, \eqref{eq: lam LoM} holds for $t=1$.
    
    Now, for $t>1$, suppose that for all agents $j$, $a_{j,t-1}=1$ if and only if $\mu_{j,t-1}\geq \mu_{t-1}^*$. By standard Gaussian-Gaussian Bayesian updating, we have $\theta|\mathcal{F}_{it} \sim N(\mu_{it},\eta_t^2)$, where
    \begin{align*}
    % \mu_{it}&=\frac{\eta_{t-1}^{-2}}{\eta_{t-1}^{-2}+\sigma_t^{-2}}\mu_{i,t-1}+\frac{\sigma_t^{-2}}{\eta_{t-1}^{-2}+\sigma_t^{-2}} x_{it}
    % \eta_t^2&=\frac{\eta_{t-1}^2}{1+\sigma_t^{-2}\eta_{t-1}^2}
    \mu_{it}=\sum_{s=1}^t \frac{\sigma_s^{-2}}{\eta_t^{-2}} x_{is} \quad \text{and} \quad \eta_t^2=\left(\sum_{s=1}^t \sigma_s^{-2}\right)^{-1}.
    \end{align*}

Conditional on $\theta$, the time-$(t-1)$ distribution of posterior means is $\mu_{j,t-1}|\theta \sim N(\theta,\eta_{t-1}^2)$. Since agents' information sets are independent conditional on $\theta$, this implies that, under agent $i$'s time-$t$ posterior belief, the marginal distribution of time-$(t-1)$ posterior means is $\mu_{j,t-1}|\mathcal{F}_{it} \sim N(\mu_{it},\eta_t^2+\eta_{t-1}^2)$. Hence, by the inductive step, $i$ takes $a_{it}=1$ if and only if
\begin{align*}
    c\leq \mathbb{E}_{it}\left[\theta+\lambda_{t-1}(\theta) \right]=\mu_{it}+\mathbb{P}_{it}(\mu_{j,t-1}\geq \mu_{t-1}^*)=\mu_{it}+\Phi\left(\frac{\mu_{it}-\mu_{t-1}^*}{\sqrt{\eta_{t-1}^2+\eta_t^2}}\right)
\end{align*}

    % =\mu_{it}+\mathbb{E}_{it}\left[\Phi\left(\frac{\theta-\mu_{t-1}^*}{\eta_{t-1}}\right)\right] \\
    % =\mu_{it}+\mathbb{E}_{it}\mathbb{E}_Z \mathbbm{1}\Big[\eta_{t-1}Z-\theta \leq -\mu_{t-1}^*\Big]=\mu_{it}+\Phi\left(\frac{\mu_{it}-\mu_{t-1}^*}{\sqrt{\eta_{t-1}^2+\eta_t^2}} \right).
% which equals \textcolor{blue}{[put indicator arg here]}
% \[
% =\mu_{it}+\int_{\theta=-\infty}^{\theta=\infty} \int_{z=-\infty}^{z=(\theta-\mu_{t-1}^*)/\eta_{t-1}} \varphi(z)\varphi_{it}(\theta) \ dz \ d\theta=\mu_{it}+\Phi\left(\frac{\mu_{it}-\mu_{t-1}^*}{\sqrt{\eta_{t-1}^2+\eta_t^2}}\right)
% \]
Let $F(\mu_{it})$ be the RHS function. Since $F$ is continuous and strictly increasing in $\mu_{it}$, and $F(c)>c$ and $F(c-1)<c$ hold, $F(\mu_{it})=c$ has a unique solution $\mu_t^*\in(c-1,c)$. Hence, there exists a unique cutoff $\mu_t^*$, defined implicitly by \eqref{eq: mu star} at time $t$, such that $i$ takes $a_{it}=1$ if and only if $\mu_{it}\geq \mu_t^*$. 

Finally, since $\mu_{it} | \theta \sim N(\theta,\eta_t^2)$, we see that $\lambda_t(\theta)=\mathbb{P}_\theta(\mu_{it}\geq \mu_t^*)=\Phi\left(\frac{\theta-\mu_t^*}{\eta_t}\right)$ and hence \eqref{eq: lam LoM} holds at time $t$, as desired.
\end{proof}

% More generally, letting $F_t(\cdot|\theta)$ denote the objective distribution of time-$t$ posteriors given $\theta$ and $G_t(\cdot|\mu_{it})$ denote $i$'s posterior over time-$t-1$ posteriors given time-$t$ information, (and assuming MLRP and symmetry), we have LoMs
% \begin{align*}
%     \lambda_t=1-F_t(\mu_t^*|\theta) \\
%     \int d\mu_t^*+1-\int dG_t(\mu_t^*)=c
% \end{align*}

\begin{lemma}
    \label{lemma: mu star ss}
    $(\mu_t^*)_t$ has a unique steady state $\mu_{ss}^*=c-1/2$. For any $t\in\mathcal{T}$, $\mu_{t+1}^*>\mu_t^*$ if and only if $\mu_t^*<\mu_{ss}^*$.
\end{lemma}

\begin{proof}
    Assume $\mu_{t-1}^*=c-(1/2)$. Since $\mu_t^*:=c-(1/2)$ solves \eqref{eq: mu star} and it has a unique solution, $c-(1/2)$ is a steady state. To see that it's the unique steady state, assume $\mu_{t-1}^*=\mu_t^*$. Then, \eqref{eq: mu star} reduces to $\mu_t^*+\Phi(0)=c$ and hence $\mu_t^*=c-(1/2)$, as desired.
    
    Next, fix any $t\in\mathcal{T}$. We have
    \begin{align*}
        \mu_{t+1}^*>\mu_t^* \Leftrightarrow\quad\Phi\left(\frac{\mu_{t+1}^*-\mu_t^*}{\sqrt{\eta_t^2+\eta_{t+1}^2}}\right)>\frac{1}{2} 
        \Leftrightarrow\quad\mu_{t+1}^*<\mu_{ss}^*. \tag{From \eqref{eq: mu star}}
    \end{align*}
    Hence, if $\mu_{t+1}^*>\mu_t^*$, $\mu_t^*<\mu_{t+1}^*<\mu_{ss}^*$. Conversely, if $\mu_{t+1}^*\leq\mu_t^*$, which is equivalent to $\mu_{t+1}^*\geq\mu_{ss}^*$, we have $\mu_t^*\geq\mu_{t+1}^*\geq\mu_{ss}^*$.
\end{proof}

% \textcolor{blue}{Replace with $\max\{c_1,c-1\}$ and $\min\{c_1,c\}$? }

\begin{lemma}
    \label{lemma: mu star limit}
    \begin{itemize}
        \item[(i)] If $\lambda_0>\frac{1}{2}$, then $\mu_t^* \uparrow \mu_\infty^*\in(\mu_1^*,\mu_{ss}^*]$.
        \item[(ii)] If $\lambda_0<\frac{1}{2}$, then $\mu_t^* \downarrow \mu_\infty^*\in[\mu_{ss}^*,\mu_1^*)$.
        \item[(iii)] $|\mu^*_{t + 1} - \mu_t^*|$ is decreasing in $t$.\footnote{Note that parts (i),(ii), and (iii) of this Lemma hold without Assumption~\ref{assumption: limit learning}.}
        \item[(iv)] $\lambda_\infty(\theta)=\mathbbm{1}[\theta \geq \mu_\infty^*]$ a.e.
    \end{itemize}
\end{lemma}

\begin{proof}
\textbf{Parts (i) and (ii).} Wlog, assume $\lambda_0 > \frac{1}{2}$. Since Lemma~\ref{lemma: mu star ss} implies that $(\mu_t^*)_t$ is a strictly increasing sequence bounded above by $\mu_{ss}^*$, it converges to $\mu_\infty^*\in(\mu_1^*,\mu_{ss}^*]$.

\textbf{Part (iii).} Note that \eqref{eq: mu star} implies $|\mu^*_{t + 1} - \mu_t^*| = \sqrt{\eta_t^2 + \eta_{t + 1}^2}\left|\Phi^{-1}(c - \mu_{t + 1}^*)\right|.$ Note that $\sqrt{\eta_t^2 + \eta_{t + 1}^2}$ is non-increasing in $t$. First, assume $\lambda_0 > \frac{1}{2}$. Then, $(c-\mu_{t+1}^*)> \frac{1}{2}$ decreasing in $t$ implies $\Phi^{-1}(c-\mu_{t+1}^*)>0$ is decreasing in $t$, and therefore $|\Phi^{-1}(c-\mu_{t+1}^*)|$ is decreasing in $t$. Hence, $|\mu^*_{t + 1} - \mu_t^*|$ is also decreasing in $t$. 

Next, assume $\lambda_0<\frac{1}{2}$. Then, $(c-\mu_{t+1}^*)< \frac{1}{2}$ increasing in $t$ implies $\Phi^{-1}(c-\mu_{t+1}^*)<0$ increasing in $t$, and therefore $|\Phi^{-1}(c-\mu_{t+1}^*)|$ is decreasing in $t$. Hence, $|\mu^*_{t + 1} - \mu_t^*|$ is also decreasing in $t$. 

\textbf{Part (iv).} Wlog, assume $\lambda_0 > \frac{1}{2}$. First, we show $(\theta-\mu_t^*)/\eta_t\to-\infty$ when $\theta<\mu_\infty^*$. Since $((\theta-\mu_t^*)/\eta_t)_t$ is decreasing and becomes negative for sufficiently large $t$, it should either converge to some negative constant or diverge to $-\infty$. Since $\lim_{t\to\infty}(\theta-\mu_t^*)<0$, $\lim_{t\to\infty}(\theta-\mu_t^*)/\eta_t=-\infty$ must hold by limit learning.

Next, we show $(\theta-\mu_t^*)/\eta_t\to\infty$ when $\theta>\mu_\infty^*$. Lemma~\ref{lemma: mu star limit} (i) implies
\begin{equation*}
    0<\frac{\theta-\mu_\infty^*}{\eta_t}<\frac{\theta-\mu_t^*}{\eta_t}. \tag{$\mu_t^* < \mu_\infty^*$}
\end{equation*}
Since $(\theta-\mu_\infty^*)/\eta_t\to\infty$ as $t\to\infty$, $(\theta-\mu_t^*)/\eta_t\to\infty$ holds.
\end{proof}

% \begin{lemma}
%     \label{lemma: monotone selection}
%     $\lambda_\infty(\theta)=\mathbbm{1}[\theta\geq \mu_\infty^*]$.
% \end{lemma}

% \begin{proof}

% \end{proof}

% \textcolor{blue}{[TODO]: For Lemmas 4-7, get rid of $\gamma_1$ (which isn't well-defined for $c_1 \notin (c-1,c)$) and reinitialize at $\gamma_2(c_1):=(\mu_2^*(c_1)-c_1)/A_2$.}

% \textcolor{blue}{[TODO]: Decide on the following. Under our new parametrization $(\Sigma,c_1)$, we have two choices.}
% \begin{enumerate}
%     \item Assume $c_1 \in (c-1,c)$. Then, all of our results and proofs go through exactly as before, replacing $\lambda_0$ with $c-c_1$ and $\gamma_1$ with $\Phi^{-1}(c-c_1)$.
%     \item Allow for any $c_1 \in \mathbb{R}$. This is more natural for the economic interpretation (there's no reason why the cost shock should be bounded in $(c-1,c)$), AND the main ideas underlying our math still hold, but it means that we'll have to restate some results (eg in Prop 2, we have to replace ``between $c-\lambda_0$ and $c-(1/2)$ with ``if $c_1<c-(1/2)$, between $\max\{c-1,c_1\}$ and $c-(1/2)$".) We'll also have to carefully alter the proofs of Lemmas 4-7, Theorem 1.(ii), and others to reinitialize the gamma process at 
%     \[
%     \gamma_2(c_1):=\frac{\mu_2^*(c_1)-c_1}{A_2}
%     \]
%     since, when $c_1 \notin (c-1,c)$, $\gamma_1$ is not well-defined and hence no longer makes sense to use as an initial condition.
% \end{enumerate} 

\begin{lemma}
\label{lem:gamma_lom}
Define $\gamma_1:=\Phi^{-1}(\lambda_0)$, and for $t\geq 2$, define
\[
\gamma_t:=\frac{\mu_t^*-\mu_{t-1}^*}{\sqrt{\eta_t^2+\eta_{t-1}^2}} \quad \text{and} \quad A_t:=\sqrt{\eta_t^2+\eta_{t-1}^2}>0
\]
Then,
\begin{itemize}
    \item[(i)] $(\gamma_t)_{t\geq 2}$ obeys the law of motion
    \begin{equation}
    \label{eq: gamma_lom}
    \tag{$\gamma$-LoM}
    A_t\gamma_t=\Phi(\gamma_{t-1})-\Phi(\gamma_t)
    \end{equation}
    \item[(ii)] If $\lambda_0>(1/2)$, then $\gamma_t \downarrow \gamma_\infty\geq 0$.
    \item[(iii)] If $\lambda_0<(1/2)$, then $\gamma_t \uparrow \gamma_\infty\leq 0$.
\end{itemize}
\end{lemma}

\begin{proof}
\textbf{Part (i).} Subtracting (\ref{eq: mu star}) at time $t-1$ from (\ref{eq: mu star}) at time $t$ yields:
\begin{align*}
    A_t\gamma_t=\mu_t^*-\mu_{t-1}^*=\Big(c-\Phi(\gamma_t)\Big)-\Big(c-\Phi(\gamma_{t-1})\Big)=\Phi(\gamma_{t-1})-\Phi(\gamma_t)
\end{align*}
\textbf{Parts (ii) and (iii).} Given the initial condition $\mu_1^*=c-\lambda_0$, defining $\gamma_1:=\Phi^{-1}(\lambda_0)$ yields $\mu_1^*=c-\Phi(\gamma_1)$. Note that for all $t \geq 2$,
\begin{align*}
    \gamma_{t-1} \geq \gamma_t \iff \Phi(\gamma_{t-1}) \geq \Phi(\gamma_t) \iff \gamma_t\geq 0 \tag{From \eqref{eq: gamma_lom}}
\end{align*}
Wlog, assume $\lambda_0>(1/2)$, then $\gamma_1>0$. Suppose for a contradiction that $\gamma_2<0$: then $\gamma_1<\gamma_2<0$, which is a contradiction since $\gamma_1>0$ by definition. Hence, $\gamma_1\geq \gamma_2\geq 0$. By induction, $(\gamma_t)_t$ is a monotone decreasing sequence bounded below by $0$, so it converges: $\gamma_t \downarrow \gamma_\infty \geq 0$, as desired. 
% The case where $\lambda_0<(1/2)$ is exactly analogous.
\end{proof}

Next, we study the function $F_A$ such that \eqref{eq: gamma_lom} implies: $\gamma_t=F_{A_t}(\gamma_{t-1})$. 

\begin{lemma}
\label{lem: F_A_properties}
For any $A>0$, let $F_A: [0,\infty) \rightarrow [0,\infty)$ be implicitly defined as follows:
% \footnote{The proof of Lemma \ref{lem: F_A_properties} can be extended to show: for any fixed $\gamma_1>0$, $F_A \downarrow 0$ uniformly on $[0,\gamma_1]$ as $A \to \infty$.}
\[
A F_A(x)=\Phi(x)-\Phi(F_A(x))
\]
Then,
\begin{itemize}
    \item[(i)] Fix any $x>0$. $F_A(x)$ is decreasing in $A$, and $F_A(x) \downarrow 0$ as $A \to \infty$.
    \item[(ii)] Fix any $A>0$. $F_A(x)$ is increasing in $x$. 
\end{itemize}
\end{lemma}

\begin{proof}
\textbf{Part (i).} Fix any $x>0$. For any $A'>A>0$, note that $A' F_{A'}(x)-AF_A(x)=\Phi(F_A(x))-\Phi(F_{A'}(x))$. Suppose for a contradiction that $F_{A'}(x)\geq F_A(x)$. Since $\Phi$ is increasing, the above implies $A' F_{A'}(x)\leq AF_A(x)$, which is a contradiction since $A'>A$. 

Next, we follow the argument of Lemma~\ref{lem:gamma_lom}(ii). Assume $x>0$, and suppose for a contradiction that $x<F_A(x)$. Then, $AF_A(x)=\Phi(x)-\Phi(F_A(x))<0$ and $A>0$ implies $F_A(x)<0<x$, a contradiction. Hence, $x \geq F_A(x)\geq 0$, which implies $\Phi(x)-\Phi(F_A(x)) \in [0,\Phi(x)-(1/2)]$ and hence
\[
F_A(x) \in \Big[0,A^{-1}(\Phi(x)-(1/2))\Big] \Rightarrow F_A(x) \downarrow 0 \quad \text{as } A\to \infty
\]
as desired.

\textbf{Part (ii).} Let $y=F_A(x)$ and use the implicit function theorem to differentiate it with respect to $x$:
\begin{align*}
    Ay'=\varphi(x)-\varphi(y)y' 
    \implies y'=\frac{\varphi(x)}{A+\varphi(y)}>0,
\end{align*}
where $\varphi$ is a standard normal PDF.
\end{proof}

% \begin{lemma}
% \label{lem: A_gamma_properties}
% For any $\gamma^*>0$, let $A_{\gamma^*}: [\gamma^*,\infty) \rightarrow [0,\infty)$ be defined as:
% \[
% A_{\gamma^*}(x):=\frac{\Phi(x)-\Phi\Big((\gamma^*+x)/2\Big)}{(\gamma^*+x)/2}
% \]
% Then, $A_{\gamma^*}(\gamma^*)=0$ and $A_{\gamma^*}$ achieves a unique global maximum at some $m_{\gamma^*}>\gamma^*$.
% \end{lemma}

% \begin{proof}
% Note that
% \begin{align*}
%     A_{\gamma^*}'(x) = \left(\frac{\gamma^*+x}{2}\right)^{-2}\left(\frac{\gamma^*+x}{2}\left[\varphi(x)-\frac{1}{2}\varphi\left(\frac{\gamma^*+x}{2}\right)\right]-\frac{1}{2}\left[\Phi(x)-\Phi\left(\frac{\gamma^*+x}{2}\right)\right]\right),
% \end{align*}
% and 
% \[
% x \mapsto \left[\varphi(x)-\frac{1}{2}\varphi\left(\frac{\gamma^*+x}{2}\right) \right]
% \]
% is monotone decreasing until after its unique root. Finally, note that
% \[
% A_{\gamma^*}'(\gamma^*)=\frac{\gamma^*\varphi(\gamma^*)}{2(\gamma^*)^2}>0
% \]
% and the conclusion follows from continuity and monotonicity of $A_{\gamma^*}'$ and the intermediate value theorem.
% \end{proof}

Next, we study the function $A_{\gamma^*}$ such that \eqref{eq: gamma_lom} implies: $F_{A_{\gamma^*}(\gamma_{t-1})}(\gamma_{t-1})=(\gamma^*+\gamma_{t-1})/2$. 

\begin{lemma}
\label{lem: A_gamma_properties}
For any $\gamma^*>0$, let $A_{\gamma^*}: [\gamma^*,\infty) \rightarrow [0,\infty)$ be defined as:
\[
A_{\gamma^*}(x):=\frac{\Phi(x)-\Phi\Big((\gamma^*+x)/2\Big)}{(\gamma^*+x)/2}
\]
Then, there exists $m > \gamma^*$ such that $A_{\gamma^*}(x) > 0$ and $A_{\gamma^*}'(x) > 0$ for all $x \in (\gamma^*, m)$.
\end{lemma}

\begin{proof}
Note that
\begin{align*}
    A_{\gamma^*}'(x) = \left(\frac{\gamma^*+x}{2}\right)^{-2}\left(\frac{\gamma^*+x}{2}\left[\varphi(x)-\frac{1}{2}\varphi\left(\frac{\gamma^*+x}{2}\right)\right]-\frac{1}{2}\left[\Phi(x)-\Phi\left(\frac{\gamma^*+x}{2}\right)\right]\right),
\end{align*}
and
\[
A_{\gamma^*}'(\gamma^*)=\frac{\gamma^*\varphi(\gamma^*)}{2(\gamma^*)^2}>0.
\]
Since $A_{\gamma^*}'$ is continuous, there exists $m > \gamma^*$ such that $A_{\gamma^*}'(x) > 0$ for all $x \in [\gamma^*, m)$. Since $A_{\gamma^*}(\gamma^*) = 0$ holds by definition, this implies $A_{\gamma^*}(x) > 0$ for all $x \in [\gamma^*, m)$.
\end{proof}

Finally, we study the function $A_{\gamma_1}$ such that \eqref{eq: gamma_lom} implies: $F_{A_{\gamma_1}(x)}(\gamma_1)=x$. 

\begin{lemma}
\label{lem: A_gamma_1_properties}
For any $\gamma_1>0$, let $A_{\gamma_1}: [0,\gamma_1] \rightarrow [0,\infty)$ be defined as
\[
A_{\gamma_1}(x):=\frac{\Phi(\gamma_1)-\Phi(x)}{x}.
\]
Then, $A_{\gamma_1}(\gamma_1)=0$ and $A_{\gamma_1}(x)$ is decreasing on $x \in [0,\gamma_1]$. Thus, for any fixed $\gamma^*\in (0, \gamma_1)$,
\[
A_{\gamma_1}(x) \uparrow A_{\gamma_1}(\gamma^*)=:\bar{A}^*>0
\]
holds as $x \downarrow \gamma^*$.
\end{lemma}

\begin{proof}
$A_{\gamma_1}(\gamma_1)=0$ follows immediately from the definition.
Note that $A_{\gamma_1}(x)$ is continuous and differentiable on $x \in [0,\gamma_1]$.
\[
A_{\gamma_1}'(x)=x^{-2}\Big[-x\varphi(x)-\Big(\Phi(\gamma_1)-\Phi(x)\Big) \Big]<0 \quad \forall x \in [0,\gamma_1]
\]
Then $A_{\gamma_1}(x) \uparrow \bar A^*$ follows from continuity of $A_{\gamma_1}(\cdot)$.
\end{proof}

With these lemmas in hand, we can understand how $(\gamma_t)_t$ changes as a function of $(A_t)_t$. 

\subsection{Theorem~\ref{thrm:char} and Corollary~\ref{cor:prop_1}} \label{appendix: thrm1}
Without loss of generality, assume $\lambda_0 > \frac{1}{2}$.

\begin{proof}[Proof of Theorem~\ref{thrm:char} (i)]
First, note that for any $\epsilon>2$, $\lim_{t\to\infty} t^{2-\epsilon}=0$. Hence, $\eta_t^{-2}=O(t^{2-\epsilon})$ implies $\lim_{t\to\infty} \eta_t^{-2}=0$. By our assumption of limit learning, the implication we want to show is vacuously true.

Hence, fix any $\epsilon \in (0,2)$, and let $p=2-\epsilon$. Suppose there exist $C,T$ such that $\eta_t^{-2} \leq Ct^p$ holds for all $t \geq T$. Then we have
\begin{align*}
    \eta_t^{-2} \leq Ct^p \implies \eta_t^2 \geq C^{-1}t^{-p} \implies \eta_t^2+\eta_{t-1}^2 \geq 2\eta_t^2 \geq 2C^{-1}t^{-p} \tag{$\eta_t^2 \leq \eta_{t-1}^2$}\\
    \implies (\eta_t^2+\eta_{t-1}^2)^{-1/2} \leq (C/2)^{1/2} t^{p/2}
\end{align*}
for all $t \geq T$. By Lemma~\ref{lemma: mu star limit} (iii), $|\mu_{t+1}^*-\mu_t^*|=\mu_{t+1}^*-\mu_t^*$ is decreasing in $t$, which implies
\begin{align*}
    \mu_{t+1}^*-\mu_t^* \leq t^{-1} \sum_{s=1}^t (\mu_{s+1}^*-\mu_s^*)=t^{-1} (\mu_{t+1}^*-\mu_1^*)\leq t^{-1}. \tag{$\mu_{t+1}^* - \mu_1^* \leq 1$}
\end{align*}
Using this inequality, we can apply the squeeze theorem and compute $\mu_\infty^*$ as follows:
\begin{align*}
    0\leq \frac{\mu_{t+1}^*-\mu_t^*}{\sqrt{\eta_t^2+\eta_{t+1}^2}}\leq \underbrace{(C/2)^{1/2} t^{p/2-1}}_{\to 0 \,\,\,\,\because\, p < 2} \\
    \implies \mu_\infty^*=c-\Phi\left(\lim_{t\to\infty}\frac{\mu_{t+1}^*-\mu_t^*}{\sqrt{\eta_t^2+\eta_{t+1}^2}}\right)=c-(1/2)
\end{align*}
and the conclusion follows from Lemma~\ref{lemma: mu star limit} (iv).
\end{proof}

\begin{proof}[Proof of Theorem~\ref{thrm:char} (ii)]
Throughout this proof, we denote $q=2+\epsilon$ and when convenient, we omit dependencies on the variables $q>2$ and (wlog) $\lambda_0>1/2$. Let ``p.m." denote ``on an interval of positive (Lebesgue) measure."

\textbf{Step 1.} First, we show the following lemma.

\begin{lemma}
\label{lem: old-thm-1(ii)}
For any $q>2$ and $\lambda_0>1/2$, there exists $\underline{C}>0$ (depending on $q$ and $\lambda_0$) such that, for any $T\geq 1$,
\[
\eta_t^{-2}\geq \underline{C}t^q \quad \forall t\geq T \implies \lambda_\infty(\theta)=NRD(\theta) \quad \text{p.m.}
\]
\end{lemma}

\begin{proof}[Proof of Lemma \ref{lem: old-thm-1(ii)}]
Let $\gamma_1:=\Phi^{-1}(\lambda_0)$ and define the following expressions:
\begin{gather*}
    L:=\frac{2^{-q}}{2^{-q}+1}, \ \ \ r:=\frac{q}{2}>1, \ \ \ S:=\sum_{t=2}^\infty t^{-r}\leq \int_{t=1}^{t=\infty} t^{-r} \ dt=\frac{1}{r-1}<\infty, \\
    K:=\left[\frac{1}{2}\frac{\varphi(\gamma_1)}{S}\right]^2, \ \ \ C:=LK, \ \ \ \underline{C}:=C^{-1}
\end{gather*}
Note that $\underline{C}$ does not depend on $t$, but does depend on $q,\lambda_0$. We begin with the case $T=1$. Fix any $\Sigma$ with induced posterior precisions $(\eta_t^{-2})_t$ such that $\eta_t^{-2} \geq \underline{C}t^q$ for all $t \geq 1$. Then the following implications hold:
\begin{align}
    \eta_t^{-2} \geq \underline{C}t^q \quad \forall t\geq 1 \notag\\
    \implies \eta_t^2 \leq Ct^{-q} \quad \forall t\geq 1 \notag\\
    \implies \eta_t^2+\eta_{t-1}^2 \leq C(t^{-q}+(t-1)^{-q}) \quad \forall t\geq 2 \notag\\
    \implies (A_t)^2:=\eta_t^2+\eta_{t-1}^2 \leq Kt^{-q} \quad \forall t\geq 2 \notag\\
    \implies A_t \leq \frac{1}{2}\frac{\varphi(\gamma_1)}{S} t^{-r} \quad \forall t\geq 2 \label{ineq: A_t}
\end{align}
where the third implication follows from
\[
\frac{C}{K}=L=\frac{2^{-q}}{2^{-q}+1} \leq \frac{t^{-q}}{t^{-q}+(t-1)^{-q}} \quad \forall \ t\geq 2.
\]
since
\begin{align*}
    \frac{d}{dt}\left(\frac{t^{-q}}{t^{-q}+(t-1)^{-q}}\right)=q(t(t-1))^{-q}\Big((t-1)^{-1}-t^{-1} \Big)\geq 0
\end{align*}
% =\Big( t^{-q}+(t-1)^{-q}\Big)^{-2}\Big((t^{-q}+(t-1)^{-q})(-qt^{-q-1})-t^{-q}\Big(-qt^{-q-1}-q(t-1)^{-q-1}\Big) \Big) \\

Next, for each $\epsilon \in (0,\gamma_1)$ and each $t\geq 2$, define:
\begin{align*}
    \Delta_t(\epsilon):=\frac{\gamma_1-\epsilon}{S}t^{-r}, \ \ \ \gamma_t(\epsilon):=\gamma_1-\sum_{s=2}^t \Delta_t(\epsilon) , \ \ \ A_t(\epsilon):=\frac{\Phi(\gamma_{t-1}(\epsilon))-\Phi(\gamma_t(\epsilon))}{\gamma_t(\epsilon)}
\end{align*}
with the initial condition $\gamma_1 (\epsilon) = \gamma_1$. Observe that $\gamma_t(\epsilon) \downarrow \epsilon$ as $t\to\infty$ and
\begin{align}
    \Phi(\gamma_{t-1}(\epsilon))-\Phi(\gamma_t(\epsilon)) &\geq \varphi(\gamma_{t-1}(\epsilon)) (\gamma_{t-1}(\epsilon) - \gamma_t(\epsilon)) \tag{$0 < \gamma_t(\epsilon) < \gamma_{t-1}(\epsilon)$ $\forall t\geq 2$}\\
    &\geq \varphi(\gamma_1) \Delta_t(\epsilon) \tag{$0 < \gamma_{t-1}(\epsilon) \leq \gamma_1$ and $\gamma_{t-1}(\epsilon) - \gamma_t(\epsilon) = \Delta_t(\epsilon)$ $\forall t\geq 2$}\\
    &\geq \frac{\gamma_t(\epsilon)}{\gamma_1} \varphi(\gamma_1) \Delta_t(\epsilon) \quad \forall t\geq 2 \tag{$\frac{\gamma_t(\epsilon)}{\gamma_1} \in (0,1)$ $\forall t\geq 2$}\\
    \implies \frac{\varphi(\gamma_1)}{\gamma_1}&\Delta_t(\epsilon)\leq A_t(\epsilon) \quad \forall t\geq 2 \label{ineq: A_t(ep)}
\end{align}
where we also use the fact that $\varphi(x)$ is decreasing for $x\geq 0$. Also note that
\[
\lim_{\epsilon \downarrow 0} \frac{\varphi(\gamma_1)}{\gamma_1}\frac{\gamma_1-\epsilon}{S} = \frac{\varphi(\gamma_1)}{S} > \frac{1}{2}\frac{\varphi(\gamma_1)}{S}
\]
Hence, we can choose $\epsilon>0$ small enough, independent of $t$, such that
\[
A_t\underbrace{\leq}_{\eqref{ineq: A_t}} \frac{1}{2} \frac{\varphi(\gamma_{1})}{S} t^{-r}\leq \frac{\varphi(\gamma_1)}{\gamma_1}\frac{\gamma_1-\epsilon}{S}t^{-r}\underbrace{\leq}_{\eqref{ineq: A_t(ep)}} A_t(\epsilon) \quad \ \forall \ t\geq 2.
\]
By Lemma~\ref{lem: F_A_properties}, we can show $\gamma_t \geq \gamma_t (\epsilon)$ for $t \geq 2$ inductively. Hence,
\[
\gamma_\infty = \lim_{t \to \infty} F_{A_t}(\gamma_{t - 1}) \geq \lim_{t \to \infty} F_{A_t(\epsilon)}(\gamma_{t - 1}(\epsilon)) = \epsilon>0
\]
where $\gamma_\infty$ is the limit of $(\gamma_t)_t$ induced by $\Sigma$.

Next, fix any $T\geq 1$, and let $\underline{C}$ be the constant derived above (which does not depend on $T$). Fix any $\Sigma$ satisfying $\eta_t^{-2}\geq \underline{C}t^q$ for all $t\geq T$, and define:
\[
\Tilde{\underline{C}}(\Sigma):=L^{-1}\left[\frac{1}{2}\frac{\varphi(\gamma_T)}{S}\right]^{-2}\leq \underline{C}. \tag{$\gamma_T\leq \gamma_1$}
\]
where we note that $0\leq\gamma_T\leq \gamma_1$ implies $\varphi(\gamma_T)\geq \varphi(\gamma_1)$. Hence, if we define $\tilde{t}:=t-(T-1)$ and $\tilde\eta_{\tilde{t}}^{-2} = \eta_t^{-2}$, $\Sigma$ satisfies
\begin{align*}
    \tilde\eta_{\tilde{t}}^{-2} \geq \underline{C}\{\tilde{t}+(T-1)\}^q\geq \underline{C}(\tilde{t})^q \geq \Tilde{\underline{C}}(\Sigma)(\tilde{t})^q \quad \forall \tilde{t}\geq 1
\end{align*}
By an exactly analogous argument as the $T=1$ case, $\gamma_\infty>0$ holds as desired.
\end{proof}

\textbf{Step 2.} To finish the argument, we want to show: for any $q>2$ and (wlog) $\lambda_0>1/2$,
\[
\eta_t^{-2}=\Omega(t^q) \implies \lambda_\infty(\theta)=NRD(\theta) \quad \text{p.m.}
\]
Fix any $q>2$, $C>0$, $T\geq 1$, and $\Sigma$. Fix any $q' \in (2,q)$, and choose $T'\geq T$ large enough such that $C(T')^{q-q'}\geq \underline{C}(q',\lambda_0),$ where $\underline{C}(q',\lambda_0)$ is the constant from Lemma \ref{lem: old-thm-1(ii)}. Then, the following implications hold: 
\begin{align*}
    \eta_t^{-2} \geq Ct^q=Ct^{q-q'}\cdot t^{q'} \quad \forall t\geq T \\
    \implies \eta_t^{-2} \geq \underline{C}(q',\lambda_0) t^{q'} \quad \forall t\geq T' \\
    \implies \lambda_\infty(\theta)=NRD(\theta) \quad \text{p.m.} \tag{From Lemma~\ref{lem: old-thm-1(ii)}}
\end{align*}
\end{proof}

\begin{proof}[Proof of Corollary \ref{cor:prop_1}]
\label{lem:pf_cor_1}
\textbf{Part (i).} If $\sigma_t=\sigma$ for all $t\geq 1$, then $\eta_t^{-2}=\sigma^{-2}t$. Take $C:=\sigma^{-2}$ and $p:=1$, and the conclusion directly follows from Theorem~\ref{thrm:char} (i).

% \underline{Part (ii)}: Although we have relaxed Assumption \ref{assumption: lim_learning}, this part follows by an analogous argument as in Proposition \ref{thrm:char}. Formally, for any fixed $\sigma>0$, note that $\eta_t^2=\sigma^2$ is constant in $t$, and the threshold update mapping $\mu(x)$ defined implicitly by 
% \[
% \mu(x)+\Phi\left(\frac{\mu(x)-x}{\sqrt{2}\sigma}\right)=c
% \]
% is a contraction mapping with unique global attactor at $c-(1/2)$. Hence, 
% \[
% \lim_{\sigma \downarrow 0} \mu_\infty^*(\sigma)=\lim_{\sigma \downarrow 0} (c-1/2)=c-1/2
% \]
% Finally, for any fixed $\theta$,
% \[
% \lim_{\sigma \downarrow 0}\lambda_\infty(\theta,\sigma)=\lim_{\sigma \downarrow 0}\Phi((\theta-(c-1/2))/\sigma)=\mathbbm{1}(\theta \geq c-1/2) \quad \text{a.s.}
% \]
% as desired.

\textbf{Part (ii).} If $\sigma_1:=\sigma$ and $\sigma_t:=\eta_{t-1}$ for all $t\geq 2$, then $\eta_t^{-2}=\sigma^{-2}2^{t-1}$ holds. Fix $q=3$ and any $\lambda_0$, and let $\underline C$ be the constant from Lemma~\ref{lem: old-thm-1(ii)} when $q = 2$. Choose $\underline{\sigma}>0$ small enough such that $\underline{\sigma}^{-2}2^{t-1}\geq \underline Ct^3 \quad \forall \ t\geq 1,$ and note that $\eta_t^{-2} \geq \underline{\sigma}^{-2}2^{t-1}$ for all $\sigma<\underline{\sigma}$. The conclusion then follows from Lemma~\ref{lem: old-thm-1(ii)}.
\end{proof}

\subsection{Proposition~\ref{prop: comp stat}} \label{appendix: comp}
\begin{proof}
    Without loss of generality, we assume $\lambda_0 > \frac{1}{2}$. 

    \textbf{Parts (i) and (ii)}:
    First, we prove Part (ii). Fix $\Sigma \geq \Sigma'$. Then it follows that $A_t := \sqrt{\eta_t^2 + \eta_{t-1}^2}\geq \sqrt{(\eta_t')^2 + (\eta_{t-1}')^2} =: A_t'$ for all $t\geq 2$. Let $(\gamma_t)_t$ and $(\gamma_t')_t$ be the induced sequences, respectively. We use induction to show that $\gamma_t\leq \gamma_t'$ for all $t\geq 1$. For the base case, note that $\gamma_1=\Phi^{-1}(\lambda_0)=\gamma_1'$. For the inductive step, assume $\gamma_t\leq \gamma_t'$. By Lemma \ref{lem: F_A_properties},
\[
\gamma_{t+1}=F_{A_{t+1}}(\gamma_t) \leq F_{A_{t+1}'}(\gamma_t)\leq F_{A_{t+1}'}(\gamma_t')=\gamma_{t+1}' \tag{$A_{t+1} \geq A_{t+1}'$ and $\gamma_t\leq \gamma_t'$}
\]
Hence, $\gamma_\infty\leq \gamma_\infty'$ holds, which implies
\[
\mu_\infty^*=c-\Phi(\gamma_\infty)\geq c-\Phi(\gamma_\infty')=(\mu_\infty^*)'
\]
as desired.

To prove Part (i), assume $\lambda_0>(1/2)$ and $\Sigma=\Sigma'$. Then, by the same inductive argument as above, we have
\begin{align*}
    \lambda_0\geq \lambda_0' \implies \gamma_1\geq \gamma_1' \implies \gamma_2=F_{A_2}(\gamma_1)\geq F_{A_2}(\gamma_1')=\gamma_2' \\
    \implies \gamma_t\geq \gamma_t' \quad \forall \ t \implies \gamma_\infty\geq \gamma_\infty' \implies \mu_\infty^*\leq (\mu_\infty^*)',
\end{align*}

        \textbf{Part (iii).} Let $\eta_t^2$ and $\tilde\eta_t^2$ denote the time-$t$ posterior variances under $\Sigma$ and $\Sigma^{s \leftrightarrow s'}$, respectively. Then, it immediately follows that $\eta_t^2 \geq \tilde\eta_t^2$ holds for all $t$. Hence, by an analogous argument as in Part (i), we have $\mu_t^* (\Sigma, \lambda_0) \geq \mu_t^*(\Sigma^{s \leftrightarrow s'}, \lambda_0)$ holds for all $t$, which implies $\mu_\infty^* (\Sigma, \lambda_0) \geq \mu_\infty^*(\Sigma^{s \leftrightarrow s'}, \lambda_0)$.
    % \begin{itemize}
    %     \item[(i)] Theorem~\ref{thrm:char} immediately implies that
    %     \[
    %     \lambda_\infty (\theta | \lambda_0) = RD(\theta) = \lambda_\infty (\theta | \lambda_0').
    %     \]

    %     \item[(ii)] We show by induction that $\mu_t^* (\lambda_0) < \mu_t^*(\lambda_0')$ holds for all $t \geq 1$. First, we have $\mu_1^*(\lambda_0) = c - \lambda_0 < c - \lambda_0' = \mu_1^*(\lambda_0')$. Now suppose that $\mu_t^*(\lambda_0) < \mu_t^*(\lambda_0')$ holds. Then, it follows that
    %     \begin{align*}
    %         &\mu_{t + 1}^*(\lambda_0) - \sqrt{\eta_t^2 + \eta_{t + 1}^2} \Phi^{-1}(c - \mu_{t + 1}^*(\lambda_0)) \\
    %         =& \mu_t^*(\lambda_0)\\
    %         <& \mu_t^*(\lambda_0') \\
    %         =& \mu_{t + 1}^*(\lambda_0') - \sqrt{\eta_t^2 + \eta_{t + 1}^2} \Phi^{-1}(c - \mu_{t + 1}^*(\lambda_0')).
    %     \end{align*}
    %     Since $\mu - \sqrt{\eta_t^2 + \eta_{t + 1}^2} \Phi^{-1}(c - \mu)$ is strictly increasing in $\mu$, the inequality above implies $\mu_{t+1}^*(\lambda_0) < \mu_{t+1}^*(\lambda_0')$.

    %     Hence, we have $\mu_\infty^* (\lambda_0) \leq \mu_\infty^*(\lambda_0')$, which implies
    %     \[
    %     \lambda_\infty (\theta | \lambda_0) = \mathbbm{1}[\theta \geq \mu_\infty^*(\lambda_0)] \geq \mathbbm{1}[\theta \geq \mu_\infty^*(\lambda_0')] = \lambda_\infty (\theta | \lambda_0').
    %     \]
    % \end{itemize}
\end{proof}

\subsection{Proposition~\ref{prop:implementable}}
We prove Proposition~\ref{prop:implementable} in two steps. First, we show that for any threshold $\mu^* \in (c - \lambda_0, c - (1/2))$, there exists a sequence $(A_t)_t$ with $A_t \downarrow 0$ as $t \to \infty$ that implements it, i.e., $\mu_\infty^* = \mu^*$. Second, we show that there exists a noise process $\Sigma=(\sigma_t)_t$ that implements the sequence $(A_t)_t$ constructed in the first step.

\begin{proof}
    Without loss of generality, we assume $\lambda_0 > \frac{1}{2}$.

    \textbf{Step 1.} Fix any $\mu^* \in (c - \lambda_0, c - (1/2))$. Our goal is to construct a sequence $(A_t)_t$ that satisfies $A_t \downarrow 0$ as $t \to \infty$, $A_{t+1}^2-A_{t+2}^2\leq A_t^2-A_{t+1}^2$ for $t \geq 2$, and $\mu_\infty^* = \mu^*$.

    By taking $t\to\infty$ for (\ref{eq: mu star}), we see that $\mu_\infty^*+\Phi(\gamma_\infty)=c$ and hence 
\[
\mu_\infty^*=\mu^* \iff \gamma_\infty=\gamma^*:=\Phi^{-1}(c-\mu^*) \in (0,\gamma_1)
\]
where $\gamma_1 := \Phi^{-1}(\lambda_0)$. Hence, it suffices to find a sequence $(A_t)_t$ whose induced $(\gamma_t)_t$ via (\ref{eq: gamma_lom}) satisfies $\gamma_t \downarrow \gamma^*$. By Lemma \ref{lem: A_gamma_properties} and Lemma~\ref{lem: A_gamma_1_properties}, we can choose $\gamma_2 \in (\gamma^*, \gamma_1)$ close enough to $\gamma^*$ such that
\[
A_2 := A_{\gamma_1}(\gamma_2) > \sqrt{2}A_{\gamma^*}(\gamma_2) > 0
\]
because i) $A_{\gamma_1} (\gamma^*) = \bar A^* >0 = A_{\gamma^*} (\gamma^*)$; ii) $A_{\gamma^*} (\gamma_2) > 0$ for $\gamma_2 \in (\gamma^*, m)$; and iii) $A_{\gamma_1} (\cdot)$ and $A_{\gamma^*} (\cdot)$ are continuous. For all $t\geq 3$, we define 
\[
A_t:=A_{\gamma^*}(\gamma_{t-1}) = \frac{\Phi(\gamma_{t-1}) - \Phi(\gamma_t)}{\gamma_t} \quad \text{and} \quad \gamma_t:=\frac{\gamma^*+\gamma_{t-1}}{2},
\]
which implies
\[
\gamma_t = \frac{\gamma_2 - \gamma^*}{2^{t - 2}} + \gamma^*.
\]
By construction, $\gamma_t \downarrow \gamma^*$ and $A_t\downarrow 0$. It remains to show that $(A_t)_t$ converges exponentially, which is a fact we use in the next step. Note that for all $t\geq 3$,
\[
\gamma_{t-1}-\gamma_t=\gamma_t-\gamma^*=4(\gamma_2-\gamma^*)2^{-t}.
\]
Since $\varphi$ is decreasing on $(0,\infty)$, we have
\begin{align*}
    &\varphi(\gamma_2)(\gamma_{t-1}-\gamma_t)\\
    \leq\varphi(\gamma_{t-1})(\gamma_{t-1}-\gamma_t)
    \leq &\Phi(\gamma_{t-1})-\Phi(\gamma_t)
    \leq \varphi(\gamma_t)(\gamma_{t-1}-\gamma_t)\tag{Mean value theorem} \\
    \leq &\varphi(\gamma^*) (\gamma_{t-1}-\gamma_t)
\end{align*}
Hence,
\begin{align*}
    \underbrace{4\frac{\varphi(\gamma_2)}{\gamma_2}(\gamma_2-\gamma^*)}_{=: c}2^{-t}\leq\frac{\Phi(\gamma_{t-1})-\Phi(\gamma_t)}{\gamma_2}\leq A_t \tag{$\gamma_t \leq \gamma_2$}\\
    \leq \frac{\Phi(\gamma_{t-1})-\Phi(\gamma_t)}{\gamma^*}\leq \underbrace{4\frac{\varphi(\gamma^*)}{\gamma^*}(\gamma_2-\gamma^*)}_{=: C} 2^{-t} \tag{$\gamma^* < \gamma_t$}
\end{align*}
as desired.

Finally, we want to show that the sequence $(A_t)_t$ constructed above satisfies $A_{t+1}^2-A_{t+2}^2\leq A_t^2-A_{t+1}^2$ for $t \geq 2$ if we take $\gamma_2$ sufficiently close to $\gamma^*$. We defined the sequence of $A_t$ so that $A_2\geq \sqrt{2}A_3$ holds. Hence, we have
\begin{align*}
    A_2\geq \sqrt{2}A_3 \implies A_2^2\geq 2A_3^2\geq 2A_3^2-A_4^2 \implies A_3^2-A_4^2\leq A_2^2-A_3^2
\end{align*}
as desired. For $t\geq 3$, note that
\begin{gather*}
    c2^{-t}\leq A_t \leq C2^{-t} \\
    \implies \left(c^2-\frac{C^2}{4}\right)2^{-2t}\leq A_t^2-A_{t+1}^2 \leq \left(C^2-\frac{c^2}{4}\right)2^{-2t}
\end{gather*}
Thus, a sufficient condition for $A_{t+1}^2-A_{t+2}^2\leq A_t^2-A_{t+1}^2$ is
\begin{align*}
    \left(C^2-\frac{c^2}{4}\right)2^{-2(t+1)} \leq \left(c^2-\frac{C^2}{4}\right)2^{-2t} \Leftrightarrow \frac{C}{c} = \frac{\varphi(\gamma^*)/\gamma^*}{\varphi(\gamma_2)/\gamma_2} \leq \sqrt{\frac{17}{8}},
\end{align*}
which is guaranteed by choosing $\gamma_2$ close enough to $\gamma^*$.

    \textbf{Step 2.} To complete the proof of Proposition~\ref{prop:implementable}, we show that there exists a noise process $\Sigma=(\sigma_t)_t$ that implements the sequence $(A_t)_t$ constructed above. First, we claim that it suffices to find a non-negative, weakly decreasing sequence $(\eta_t^2)_t$ of posterior variances. This is true because given any such $(\eta_t^2)_t$, the sequence $(\eta_t^{-2})_t$ is non-negative and weakly increasing. Thus, $(\eta_t^2)_t$ is implemented by the noise process $\sigma_1^{-2}:=\eta_1^{-2},$ and $\sigma_t^{-2}=\eta_t^{-2}-\eta_{t-1}^{-2}$ for $t\geq 2$. Thus, it suffices to find a non-negative and weakly decreasing sequence $(\eta_t^2)_{t\geq 1}$ such that $A_t^2=\eta_t^2+\eta_{t-1}^2$ for $t\geq 2$. Note that any choice of $\eta_1^2$ pins down a candidate solution $(\eta_t^2)_{t\geq 1}$ by $\eta_{t+1}^2 = A_{t+1}^2 - \eta_t^2$. We therefore proceed by stating the implied constraints on $\eta_1^2$.

\textbf{Step 2.1.}
Let us begin by establishing the formula
\[
\eta_t^2(\eta_1^2):=(-1)^{t+1}\eta_1^2+\sum_{s=2}^t (-1)^{t+s} A_s^2
\]
This follows by induction. For the base case, note that $\eta_2^2=A_2^2-\eta_1^2$. For the inductive step, suppose the formula holds for $\eta_t^2$. Then,
\begin{align*}
    \eta_{t+1}^2(\eta_1^2)=A_{t+1}^2-\eta_t^2(\eta_1^2)&=A_{t+1}^2+(-1)^{(t+1)+1}\eta_1^2+\sum_{s=2}^t (-1)^{(t+1)+s} A_s^2 \\
    &=(-1)^{t+2}\eta_1^2+\sum_{s=2}^{t+1} (-1)^{(t+1)+s} A_s^2
\end{align*}
since $2(t+1)$ is even.

\textbf{Step 2.2.} Next, we pin down the implied constraints on $\eta_1^2$. The base case constraints are:
\begin{align*}
    0\leq A_2^2-\eta_1^2 \leq \eta_1^2 \iff
    L_1:=\frac{A_2^2}{2}\leq \eta_1^2 \leq A_2^2=:U_1
\end{align*}
For every $t\geq 2$, we have the constraint $0\leq \eta_{t+1}^2\leq \eta_t^2.$ We can split this set of constraints into two cases.

\underline{Case 1.}
If $t$ is even, then the time-$t$ constraint from above is
\begin{align*}
    \sum_{s=2}^{t+1}(-1)^{t+s} A_s^2 \leq \eta_1^2 \leq \sum_{s=2}^t (-1)^{t+s} A_s^2-\frac{A_{t+1}^2}{2} 
\end{align*}
which is equivalent to
\[
L_t:=(A_2^2-A_3^2)+\cdots+(A_t^2-A_{t+1}^2) \leq \eta_1^2 \leq (A_2^2-A_3^2)+\cdots+(A_t^2-(A_{t+1}^2/2))=:U_t
\]

\underline{Case 2.} If $t$ is odd, then the time-$t$ constraint from above is:
\[
\sum_{s=2}^t (-1)^{(t+1)+s}A_s^2+\frac{A_{t+1}^2}{2}\leq \eta_1^2\leq \sum_{s=2}^{t+1} (-1)^{(t+1)+s} A_s^2
\]
which is equivalent to:
\begin{align*}
    L_t:=(A_2^2-A_3^2)+\cdots+(A_{t-1}^2-A_t^2)+(A_{t+1}^2/2)\leq \eta_1^2 \\
    \leq (A_2^2-A_3^2)+\cdots+(A_{t-1}^2-A_t^2)+A_{t+1}^2=:U_t
\end{align*}
Note that for odd $t$, $L_t\geq L_{t-1}$ and for even $t$, $U_t\leq U_{t-1}$. Hence, only the lower bounds for odd $t$ and the upper bounds for even $t$ have bite.

\textbf{Step 2.3.} Finally, we show that the previously derived properties of $(A_t)_t$ imply that there exists $\eta_1^2$ satisfying the above constraints, and hence pins down a solution $\Sigma$. First, note that $(L_t)_{t \text{ odd}}$ are increasing, since for any odd $t \geq 1$,
\begin{align*}
    L_t \leq L_{t+2} \iff \frac{A_{t+1}^2}{2} \leq A_{t+1}^2-A_{t+2}^2+\frac{A_{t+3}^2}{2} \\
    \iff A_{t+2}^2-A_{t+3}^2\leq A_{t+1}^2-A_{t+2}^2
\end{align*}
which holds by Step 1. An analogous argument implies that $(U_t)_{t \text{ even}}$ are decreasing. Second, note that $A_t^2$ decreasing implies $L_t\leq U_{t+1}$ for any odd $t$. Hence, by the nested intervals theorem, we have:
\[
\sup_{t \text{ odd}} L_t=\eta_1^2=\sum_{s=1}^\infty (A_{2s}^2-A_{2s+1}^2)=\inf_{t \text{ even}} U_t
\]
and this sum converges because 
\begin{gather*}
    \left(c^2-\frac{C^2}{4}\right)2^{-4s}\leq A_{2s}^2-A_{2s+1}^2 \leq \left(C^2-\frac{c^2}{4}\right)2^{-4s} \\
    \Rightarrow \frac{1}{15}\left(c^2-\frac{C^2}{4}\right) \leq \sum_{s=1}^\infty (A_{2s}^2-A_{2s+1}^2) \leq \frac{1}{15} \left(C^2-\frac{c^2}{4}\right).
\end{gather*}
Note that $\eta_1^2$ is positive because $C/c \leq \sqrt{17/8}$ implies $c^2 - C^2/4 > 0$.
\end{proof}

\subsection{Proposition~\ref{prop: dynamics}}\label{appendix: a5}
\begin{proof}[Proof]
\label{lem:pf_prop_dynamics}
Without loss of generality, we assume $\lambda_0 > \frac{1}{2}$. \textbf{Part (i).} Fix any $\theta \in (c-\lambda_0,c-(1/2))$. For any $\sigma>0$, define $T(\sigma) \in \mathbb{N}$ implicitly as $\mu_{T(\sigma)}^*\leq \theta<\mu_{T(\sigma)+1}^*$. Note that $T(\sigma)$ is well-defined because for any fixed $\sigma>0$, $\mu_t^* \uparrow c-(1/2)$.

First, we show that $T(\sigma)\to\infty$ as $\sigma \downarrow 0$. Notice that $T(\sigma)$ is weakly decreasing in $\sigma$ because, by the proof of Proposition \ref{prop: comp stat}(ii), for any $\sigma' > \sigma$, $\theta \geq \mu^*_{T(\sigma')}(\sigma') > \mu^*_{T(\sigma')}(\sigma)$. Hence, $T(\sigma) \geq T(\sigma')$ by the definition of $T(\sigma)$. Hence, $T(\sigma)$ either converges to some constant or diverges to $\infty$ as $\sigma \downarrow 0$. Suppose for a contradiction that $\lim_{\sigma\downarrow 0}T(\sigma) = T$ holds for some $T > 0$. Since $T(\sigma) \uparrow T$ as $\sigma \downarrow 0$, for any $\sigma$, we have $\mu^*_{T + 1}(\sigma) > \mu^*_{T(\sigma) + 1}(\sigma) > \theta,$ which implies $\lim_{\sigma\downarrow 0}\mu^*_{T + 1}(\sigma) \geq \theta$. This contradicts $\lim_{\sigma\downarrow 0}\mu^*_{T + 1}(\sigma) = c - \lambda_0$.

% \footnote{See Lemma~\ref{lemma: fixed tau} for the proof of $\lim_{\sigma\downarrow 0}\mu^*_{T + 1}(\sigma) = c - \lambda_0$.}

Next, fix any $\alpha \in (0,1)$. For any $0\leq t\leq\alpha T(\sigma)$, we have:
\begin{align*}
    \theta-\mu_t^*>\mu_{T(\sigma)}^*-\mu_t^* &=\sum_{s=t}^{T(\sigma)-1} (\mu_{s+1}^*-\mu_s^*)=\sum_{s=t}^{T(\sigma)-1}  \sigma\sqrt{\frac{1}{s+1}+\frac{1}{s}} \Phi^{-1}(c-\mu_{s+1}^*) \\
    \implies \lambda_t&=\Phi\left(\frac{\sqrt{t}(\theta-\mu_t^*)}{\sigma}\right)>\Phi\left(\sum_{s=t}^{T(\sigma)-1}\sqrt{\frac{t}{s+1}+\frac{t}{s}}\Phi^{-1}(c-\mu_{s+1}^*)\right) \\
    &\geq \Phi\left((T(\sigma)-t)\sqrt{\frac{t}{T(\sigma)}+\frac{t}{T(\sigma)-1}}\Phi^{-1}(c-\theta)\right) \\
    &\geq \Phi\left((1-\alpha)T(\sigma)\sqrt{\frac{1}{T(\sigma)}+\frac{1}{T(\sigma)-1}}\Phi^{-1}(c-\theta)\right) \to 1 \quad \text{as} \quad \sigma \downarrow 0.
\end{align*}
Since our lower bound on $\lambda_t$ holds uniformly for $t \in [0,\alpha T(\sigma)]$, for any $\epsilon>0$, there exists $\underline{\sigma}_\alpha>0$ such that for all $\sigma<\underline{\sigma}_\alpha$ and all $t \in [0,\alpha T(\sigma)]$, $\lambda_t\geq 1-\epsilon$.

Finally, fix any $\beta \in (1,\overline{\beta})$, where $\overline{\beta}:=c+(1/2)-\theta-\delta \quad\text{and}\quad \delta=(c-(1/2)-\theta)/2$, and note that, for any $t \in [\beta T(\sigma),\infty)$,
\begin{align*}
    \theta-\mu_t^*\leq \theta-\mu_{\beta T(\sigma)}^*&<\mu_{T(\sigma)+1}^*-\mu_{\beta T(\sigma)}^*=-\sum_{s=T(\sigma)+1}^{\beta T(\sigma)-1} \sigma \sqrt{\frac{1}{s+1}+\frac{1}{s}}\Phi^{-1}(c-\mu_{s+1}^*) \\
    \implies \lambda_t&<\Phi\left(-\sum_{s=T(\sigma)+1}^{\beta T(\sigma)-1} \sqrt{\frac{t}{s+1}+\frac{t}{s}}\Phi^{-1}(c-\mu_{s+1}^*)\right) \\
    &<\Phi\left(-\Big( (\beta-1) T(\sigma)-1\Big) \sqrt{1+\frac{\beta T(\sigma)}{\beta T(\sigma)-1}}\Phi^{-1}(c-\mu_{\beta T(\sigma)}^*)\right) \\
    &<\Phi\left(-\Big( (\beta-1) T(\sigma)-1\Big) \sqrt{1+\frac{\beta T(\sigma)}{\beta T(\sigma)-1}}\Phi^{-1}(c-\theta-(\overline{\beta}-1))\right) \\
    &\to 0 \quad \text{as } \sigma \to 0
\end{align*}
where the inequality follows from
\begin{align*}
    \mu_{\beta T(\sigma)}^*-\mu_{T(\sigma)} &=\sum_{s=T(\sigma)}^{\beta T(\sigma)-1} (\mu_{s+1}^*-\mu_s^*) \leq \sum_{s=T(\sigma)}^{\beta T(\sigma)-1} \frac{1}{s}
    \\ &\leq (\beta-1) T(\sigma) \frac{1}{T(\sigma)}=\beta-1 \\
    \implies c-\mu_{\beta T(\sigma)}^* &\geq c-\theta-(\beta-1)>c-\theta-(\overline{\beta}-1)>\frac{1}{2} \\
    \implies \Phi^{-1}(c-\mu_{\beta T(\sigma)}^*)&>\Phi^{-1}(c-\theta-(\overline{\beta}-1))>0.
\end{align*}
Since this upper bound on $\theta-\mu_t^*$ holds uniformly for $t \in [\beta T(\sigma),\infty)$, there exists $\underline{\sigma}_\beta>0$ such that for all $\sigma<\underline{\sigma}_\beta$ and all $t \in [\beta T(\sigma),\infty)$, $\lambda_t<\epsilon$.

Finally, define $\underline{\sigma}:=\min\{\underline{\sigma}_\alpha,\underline{\sigma}_\beta\}$.

\textbf{Part (ii).} For any $t \geq 1$, we have
\begin{align*}
    |\lambda_{t + 1} - \lambda_t| &= \left| \Phi\bigg(\cfrac{\theta - \mu^*_{t+1}}{\sigma/\sqrt{t+1}}\bigg)  - \Phi\bigg(\cfrac{\theta - \mu^*_t}{\sigma/\sqrt{t}}\bigg) \right|\\
    &\leq \frac{1}{\sqrt{2\pi}\sigma} \left| (\sqrt{t + 1} - \sqrt{t})\theta - \sqrt{t + 1}\mu_{t + 1}^* + \sqrt{t} \mu_t^* \right| \tag{$\Phi'(\cdot) \leq 1/\sqrt{2\pi}$}\\
    &\leq \frac{\sqrt{t + 1} - \sqrt{t}}{\sqrt{2\pi}\sigma}|\theta| + \frac{\sqrt{t + 1}\mu_{t + 1}^* - \sqrt{t} \mu_t^*}{\sqrt{2\pi}\sigma}\\
    &= \frac{\sqrt{t + 1} - \sqrt{t}}{\sqrt{2\pi}\sigma}|\theta| + \frac{\sqrt{t + 1}}{\sqrt{2\pi}\sigma}(\mu_{t + 1}^* - \mu_t^*) + \frac{(\sqrt{t + 1} - \sqrt{t})\mu_t^*}{\sqrt{2\pi}\sigma}\\
    &<  \frac{\sqrt{t + 1} - \sqrt{t}}{\sqrt{2\pi}\sigma}|\theta| + \frac{\sqrt{t + 1}}{\sqrt{2\pi}\sigma t} + \frac{c (\sqrt{t + 1} - \sqrt{t})}{\sqrt{2\pi}\sigma} \tag{$\mu_t^* < c$ and $\mu_{t + 1}^* - \mu_t^* < t^{-1}$}\\
    &\leq \frac{1}{\sqrt{2\pi}\sigma} \left\{ (\sqrt{2} - 1) (|\theta| + c) + \sqrt{2} \right\}.
\end{align*}

% \footnote{See Lemma~\ref{lem:pf_prop1} for the proof.} 
Hence, for any $\epsilon$, there exists $\overline \sigma > 0$ such that $|\lambda_{t + 1} - \lambda_t| < \epsilon$ holds for all $t \geq 1$ and $\sigma > \overline \sigma$. Note that this argument does not depend on the value of $\theta$ as long as $|\theta| < \infty$ holds.
\end{proof}

\setstretch{1}
\setlength{\bibsep}{0pt}
\bibliography{WorksCited}

\clearpage

\titleformat{\section}
		{\bfseries\center\scshape}     
         {Appendix \thesection:}% the label and number
        {0.5em}% space between label/number and subsection title
        {}% formatting commands applied just to subsection title
        []% punctuation or other commands following subsection title

\newpage

\newpage

\clearpage 
\appendix 
\titleformat{\section}
		{\normalsize\bfseries\center\scshape}     
         {Appendix \thesection:}% the label and number
        {0.5em}% space between label/number and subsection title
        {}% formatting commands applied just to subsection title
        []% punctuation or other commands following subsection title
\renewcommand{\thesection}{\Roman{section}}

\setcounter{page}{1}
{\renewcommand{\thefootnote}{\fnsymbol{footnote}}
\begin{center}
    \large{\textbf{\MakeUppercase{Appendix to 
    `Inertial Coordination Games'}}} \\
    \normalsize{{\textbf{Andrew Koh\footnote[1]{Department of Economics, MIT, \href{mailto:ajkoh@mit.edu}{ajkoh@mit.edu}} \quad  Ricky Li\footnote[2]{Department of Economics, MIT, \href{mailto:sanguanm@mit.edu}{rickyli@mit.edu}}} \quad \textbf{Kei Uzui}\footnote[2]{Department of Economics, MIT, \href{mailto:sanguanm@mit.edu}{kuzui@mit.edu}} \\
    \small{FOR ONLINE PUBLICATION ONLY}}}
\end{center}}

\section{Small Lags}
\label{app:small lags}
% [Ricky to make the point in main text that Prop4(1) shows that the contemporaneous model is uninteresting for us, hence framing non-contemporaneous incentives as a \textit{nontriviality} assumption, and also not descriptively accurate, as it implies that learning speed doesn't affect limit outcomes, which is precisely the opposite of the stylized fact that we set out to show.]

% [Parts 2 and 3 formalize that our model's results are invariant to ``the size of the lag."]
Consider $\mathcal{T} = \{1,2,\dots \}$ and fix a learning process $\Sigma = (\sigma_t^2)_{t \in \mathcal{T}}$. Now we define an alternative environment with $\mathcal{T}^{(n)} := \left\{\frac{1}{n}, \frac{2}{n}, \dots \right\}$ and $\Sigma^{(n)} := ((\sigma_t^{(n)})^2)_{t \in \mathcal{T}^{(n)}}$ where $\sigma_t^{(n)} = \sqrt{n}\sigma_s$ when $t \in (s - 1, s] $. Note that $\sigma_t^{(n)}$ is constant in $t$ for all $t \in (s - 1, s]$. Then, the posterior variance is given by
\[
(\eta_t^{(n)})^2 = \left( \sum_{k \in \mathcal{T}^{(n)},k \leq t} \frac{1}{(\sigma_k^{(n)})^2} \right)^{-1}.
\]
Note that for $t \in \mathbb{Z}$ and letting $s(k)=\text{ceil}(k)$, we may write:
\begin{align*}
    (\eta_t^{(n)})^2 = \left( \sum_{k \in \mathcal{T}^{(n)},k \leq t} \frac{1}{n\sigma_{s(k)}^2} \right)^{-1} \\
    =n\left(\sum_{k \in \mathcal{T}^{(n)},k \leq t} \frac{1}{\sigma_{s(k)}^2} \right)^{-1} \\
    =n\left(n \frac{1}{\sigma_1^2}\cdots+n\frac{1}{\sigma_t^2} \right)^{-1}=\eta_t^2
\end{align*}

Using this, we have the following indifference condition and the aggregate time-$t$ play:
\[
\mu_{t + 1/n}^{(n)*} + \Phi\left( \frac{\mu_{t + 1/n}^{(n)*} - \mu_t^{(n)*}}{\sqrt{(\eta_{t + 1/n}^{(n)})^2 + (\eta_t^{(n)})^2}} \right) = c,\quad \lambda_t^{(n)}(\theta) = \Phi\left( \frac{\theta - \mu_t^{(n)*}}{\eta_t^{(n)}} \right)
\]
where $\mu_t^{(n)*}$ is time-$t$ threshold belief for $t \in \mathcal{T}^{(n)}$. Define $\lambda^{(n)}_\infty (\theta) := \lim_{t \to \infty} \lambda_t^{(n)}(\theta)$.

\begin{proposition} \label{prop:smallinertia}
The following statements are true: 
\begin{enumerate}
    \item[P1.]  Consider the model with \textit{contemporaneous incentives}, in which $U(1,t)=\theta+\lambda_t$. Then, the unique strategy that survives iterated deletion of strictly dominated strategies (IDSDS) is the symmetric cutoff strategy $\mu_t^*=c-(1/2)$, that is, everyone plays action $1$ if and only if her posterior mean is greater than $c -(1/2)$. 
    
    \item[P2.]  For each $n \geq 1$ and $t \geq 1$, there exists $\epsilon(n,t)>0$ such that $\mu_t^{(n)*}$ is an $\epsilon(n,t)$-equilibrium of the contemporaneous game. Furthermore, for each fixed $t\geq 1$, $\lim_{n \to \infty} \epsilon(n,t)=0$.
    \item[P3.] Fix any $(\lambda_0,\Sigma,c)$. For any $n\geq 1$, the following hold:
    \begin{itemize}
    \item[(i)] For all $\epsilon > 0$,
    \[
    \eta^{-2}_t = O(t^{2 - \epsilon}) \implies \lambda_\infty^{(n)}(\theta)=RD(\theta) \quad \text{a.e.} 
    \]
    \item[(ii)] For all $\epsilon > 0$, 
    \begin{align*}
        \eta^{-2}_t = \Omega(t^{2 + \epsilon}) \implies &\lambda_\infty^{(n)}(\theta)=NRD(\theta) \\ &\text{for some positive measure interval $(\underline \theta, \overline \theta)$.}
    \end{align*}
\end{itemize}
\end{enumerate}
\end{proposition}
%\textbf{Part 1}: Consider the model with \textit{contemporaneous incentives}, in which $U(1,t)=\theta+\lambda_t$. Then, the unique symmetric threshold strategy is $\mu_t^*=c-(1/2)$.

%\textbf{Part 2}: For each $n \geq 1$ and $t \geq 1$, there exists $\epsilon(n,t)>0$ such that $\mu_t^{(n)*}$ is an $\epsilon(n,t)$-equilibrium of the contemporaneous game. Furthermore, for each fixed $t\geq 1$, $\lim_{n \to \infty} \epsilon(n,t)=0$.

%\textbf{Part 3}: Fix any $(\lambda_0,\Sigma,c)$. For any $n\geq 1$, the following hold:
%    \begin{itemize}
%    \item[(i)] For all $\epsilon > 0$,
    %\[    \eta^{-2}_t = O(t^{2 - \epsilon}) \implies \lambda_\infty^{(n)}(\theta)=RD(\theta) \quad \text{a.e.} \]
    %\item[(ii)] For all $\epsilon > 0$, 
    %\begin{align*}
    %    \eta^{-2}_t = \Omega(t^{2 + \epsilon}) \implies &\lambda_\infty^{(n)}(\theta)=NRD(\theta) \\ &\text{for some positive measure interval $(\underline \theta, \overline \theta)$.}
    %\end{align*}
%\end{itemize}

\begin{proof} We show each part in turn. \\

    \underline{\textbf{Part 1.}} Fix any $t \in \mathcal{T}$. Define $b(x)$ as the unique value of $\mu$ solving the following equation:
    \[
    \mu + \Phi\left( \frac{\mu - x}{\sqrt{2}\eta_t} \right) = c.
    \]
    We show by induction that a strategy $s^k$ surviving $k$ rounds of IDSDS satisfies
    \[
    s^k(\mu_{it}) = \begin{cases}
        1 &\text{if $\mu_{it} > b^{k - 1} (c)$} \\
        0 &\text{if $\mu_{it} < b^{k - 1} (c - 1)$},
    \end{cases}
    \]
    where $\mu_{it}$ denotes $i$'s posterior mean at time $t$, and $b^0(x) = x$ holds. First, the statement holds for $k = 1$ because it is strictly dominant to play $1$ ($0$) when $\mu_{it} \in (c, + \infty)$ ($\mu_{it} \in (- \infty, c - 1)$). Next suppose that the statement is true for $k$. Since agents know that others play $1$ whenever their posterior mean is greater than $b^{k - 1}(c)$, agent $i$'s expected payoff of playing $1$ is at least
    \[
    \mu_{it} + \Pr_{it} (\mu_{jt} > b^{k - 1}(c)) = \mu_{it} + \Phi\left( \frac{\mu_{it} - b^{k - 1}(c)}{\sqrt{2}\eta_t} \right).
    \]
    This implies that any strategy surviving $k$ rounds such that $s^k(\mu_{it}) = 0$ for $\mu_{it} \in (b(b^{k - 1}(c)),b^{k - 1}(c)]$ is eliminated in the $(k + 1)$-th round. By an analogous argument, $s^k(\mu_{it}) = 1$ for $\mu_{it} \in [b^{k - 1}(c - 1),b(b^{k - 1}(c - 1)))$ is eliminated. Hence, the statement is true for $k + 1$. Finally, since $\lim_{k \to \infty} b^k(c) = \lim_{k \to \infty} b^k(c - 1) = c - (1/2)$ holds, a cutoff strategy with threshold $c - (1/2)$ is the unique strategy surviving IDSDS.

    \underline{\textbf{Part 2.}} Fix any $n$ and $t \in \mathcal{T}^{(n)}$. Without loss of generality, assume $\lambda_0 > \frac{1}{2}$ so that $\mu_t^{(n)*} \uparrow \mu_\infty^{(n)*}$. Suppressing dependencies on these variables, define
    \[
    x_1(\theta)=\frac{\theta - \mu_t^{(n)*}}{\eta_t^{(n)}} \quad \text{and} \quad x_2(\theta)=\frac{\theta - \mu_{t-1/n}^{(n)*}}{\eta_{t-1/n}^{(n)}}
    \]
    and $x(\theta) = \min\Big\{|x_1(\theta)|,|x_2(\theta)|\Big\} $. Let $\overline{c}=\max\{|c|,|c-1|\}$ and observe that
    \begin{align*}
        \left| \lambda_t^{(n)}(\theta) - \lambda_{t - 1/n}^{(n)}(\theta) \right| &= \left| \Phi\left( x_1(\theta) \right) - \Phi\left( x_2(\theta)\right) \right| \\
        &\leq \varphi(x(\theta)) \left| x_1(\theta) - x_2(\theta) \right|\\
        &\leq |\theta|\varphi(x(\theta)) \left(\frac{1}{\eta_t^{(n)}} - \frac{1}{\eta_{t - 1/n}^{(n)}} \right) + \frac{1}{\sqrt{2\pi}}  \frac{|\eta_{t - 1/n}^{(n)}\mu_t^{(n)*} - \eta_t^{(n)}\mu_{t - 1/n}^{(n)*}|}{\eta_t^{(n)}\eta_{t - 1/n}^{(n)}}  \tag{$\varphi(\cdot) \leq (\sqrt{2\pi})^{-1}$}\\
        &\leq |\theta|\varphi(x(\theta)) \left(\frac{1}{\eta_t^{(n)}} - \frac{1}{\eta_{t - 1/n}^{(n)}} \right) + \left|\frac{\mu_t^{(n)*} -  \mu_{t - 1/n}^{(n)*}}{\sqrt{2\pi} \eta_t^{(n)}}\right| + \left|\frac{\eta_{t - 1/n}^{(n)} - \eta_t^{(n)}}{\sqrt{2\pi} \eta_t^{(n)} \eta_{t - 1/n}^{(n)}} \mu_{t - 1/n}^{(n)*}\right| \\
        &\leq |\theta|\varphi(x(\theta))\left(\frac{1}{\eta_t^{(n)}} - \frac{1}{\eta_{t - 1/n}^{(n)}} \right) + \frac{1}{\sqrt{2\pi} \eta_t^{(n)} (nt - 1)}+\frac{\eta_{t - 1/n}^{(n)} - \eta_t^{(n)}}{\sqrt{2\pi} \eta_t^{(n)} \eta_{t - 1/n}^{(n)}} |\mu_{t - 1/n}^{(n)*}| \tag{$\mu_t^{(n)*} -  \mu_{t - 1/n}^{(n)*} \leq (nt - 1)^{-1}$ by \eqref{eq: bound on difference}}\\
        &=\left(|\theta| \varphi(x(\theta))+\frac{1}{\sqrt{2\pi}}|\mu_{t-1/n}^{(n)*}|\right)\left(\frac{1}{\eta_t^{(n)}} - \frac{1}{\eta_{t - 1/n}^{(n)}} \right)+\frac{1}{\sqrt{2\pi} \eta_t^{(n)} (nt - 1)} \\
        &\leq \max\left\{ |\mu_t^{(n)*}| + \eta_t^{(n)}+\frac{\overline{c}}{\sqrt{2\pi}}, |\mu_{t - 1/n}^{(n)*}| + \eta_{t - 1/n}^{(n)}+\frac{\overline{c}}{\sqrt{2\pi}} \right\} \left(\frac{1}{\eta_t^{(n)}} - \frac{1}{\eta_{t - 1/n}^{(n)}} \right) \\
        &+ \frac{1}{\sqrt{2\pi} \eta_t^{(n)} (nt - 1)} \\
        &\leq \left(2\overline{c}+\eta_1\right)\left(\frac{1}{\eta_t} - \frac{1}{\eta_{t - 1/n}^{(n)}} \right) +\frac{1}{\sqrt{2\pi} \eta_t (nt - 1)}:= \epsilon(n,t),
    \end{align*}
    where the last inequality follows from
    \begin{align*}
        |\theta| \varphi(|x_1(\theta)|) &= \left|\mu_t^{(n)*} + \eta_t^{(n)} x_1(\theta)\right| \varphi(|x_1(\theta)|) \\
        &\leq \left(|\mu_t^{(n)*}| + \eta_t^{(n)} |x_1(\theta)|\right) \varphi(|x_1(\theta)|) \\
        &\leq |\mu_t^{(n)*}| + \eta_t^{(n)} \tag{$|a|\varphi(|a|) \leq \varphi(1) \leq 1$}
    \end{align*}
    and $|\theta| \varphi(|x_2(\theta)|) \leq |\mu_{t - 1/n}^{(n)*}| + \eta_{t - 1/n}^{(n)}$ by analogous arguments. Finally, note that for any fixed $t\geq 1$,
    \[
    \eta_{t-1/n}^{(n)} \to \eta_t
    \]
    since
    \[
    \Big(\eta_{t-1/n}^{(n)}\Big)^{-2}=\Big(\eta_{t}^{(n)}\Big)^{-2}-n^{-1}\sigma_t^{-2}=\eta_t^{-2}-n^{-1}\sigma_t^{-2}
    \]
    and
    \[
    \frac{1}{\sqrt{2\pi} \eta_t (nt - 1)} \to 0
    \]
    as $n\to\infty$.
    % \textbf{Uniform bound}: Fix $n\geq 1$ and $t\geq 1$, and (suppressing dependencies on these variables) define
    % \[
    % x_1(\theta)=\frac{\theta - \mu_t^{(n)*}}{\eta_t^{(n)}} \quad \text{and} \quad x_2(\theta)=\frac{\theta - \mu_{t-1/n}^{(n)*}}{\eta_{t-1/n}^{(n)}}
    % \]
    % Consider the following tighter Lipschitz bound: 
    % \begin{align*}
    %     \left| \lambda_t^{(n)}(\theta) - \lambda_{t - 1/n}^{(n)}(\theta) \right| &=\Big|\Phi(x_1(\theta))-\Phi(x_2(\theta))\Big| \\
    %     &\leq \varphi\Big(\min\Big\{|x_1(\theta)|,|x_2(\theta)|\Big\}\Big)\Big|x_1(\theta)-x_2(\theta)\Big|
    % \end{align*}
    
    % Now, fix any $n\geq 1$ and $t\geq 1$. 
    Hence, the above convergence is \textit{uniform in $\theta$}; i.e. there exists $\epsilon(n,t)$ with $\epsilon(n,t) \downarrow 0$ as $n \to \infty$ for any fixed $t\geq 1$, such that
    \[
    \sup_{\theta \in \mathbb{R}} \left| \lambda_t^{(n)}(\theta) - \lambda_{t - 1/n}^{(n)}(\theta) \right|\leq \epsilon(n,t)
    \]
    Hence, the time-$t$ strategy ``$a_{it}=1$ iff $\mu_{it}\geq \mu_t^{(n)*}$" is an $\epsilon(n,t)$-equilibrium at time $t$, since we may bound the time-$t$ expected payoff of using this strategy in the contemporaneous game \textit{uniformly} across players' time-$t$ information sets $\mu_{it}$: if $\mu_{it}\geq \mu_t^{(n)*}$,
    \begin{align*}
        \mathbb{E}_{it}\Big[\theta+\lambda_{t}^{(n)}(\theta) \Big]=\mathbb{E}_{it}\Big[\theta+\lambda_{t-1/n}^{(n)}(\theta) \Big]+\mathbb{E}_{it}\Big[ \lambda_{t}^{(n)}(\theta)-\lambda_{t-1/n}^{(n)}(\theta)\Big] \\
        \geq c-\mathbb{E}_{it}\Big[ \lambda_{t-1/n}^{(n)}(\theta)-\lambda_{t}^{(n)}(\theta)\Big] \\
        \geq c-\mathbb{E}_{it}\Big| \lambda_{t-1/n}^{(n)}(\theta)-\lambda_{t}^{(n)}(\theta)\Big| \\
        \geq c-\epsilon(n,t) \quad \forall \mu_{it} \in \mathbb{R}
    \end{align*} 
    and an exactly analogous argument holds for the case $\mu_{it}<\mu_t^{(n)*}$.
    
    \underline{\textbf{Part 3 (i).}} Without loss of generality, assume $\lambda_0 > \frac{1}{2}$ so that $\mu_t^{(n)*} \uparrow \mu_\infty^{(n)*}$. By definition, $\eta^{-2}_t = O(t^{2 - \epsilon})$ means that there exist $C$ and $T$ such that $\eta_t^{-2} \leq C t^{2 - \epsilon}$ for all $t \in\{T, T+1, \dots\}$. Since $(\eta_t^{(n)})^{-2} = \eta_t^{-2}$ for $t \in \{1,2,\dots \}$, we have $(\eta_t^{(n)})^{-2} \leq C t^{2 - \epsilon}$ for all $t \in\{T, T+1, \dots\}$. Hence,
    \[
    (\eta_t^{(n)})^2 + (\eta_{t - 1/n}^{(n)})^2 \geq 2 (\eta_t^{(n)})^2 \geq 2 C^{-1} t^{-(2 - \epsilon)} \quad\text{for $t \in \{T, T+1,\dots \}$}
    \]
    Furthermore, since $\mu_t^{(n)*} - \mu_{t - 1/n}^{(n)*}$ is decreasing in $t$, we have
    \begin{equation} \label{eq: bound on difference}
        \mu_t^{(n)*} - \mu_{t - 1/n}^{(n)*} \leq \frac{1}{nt - 1} \sum_{s \in \mathcal{T}^{(n)}, s \leq t - 1/n} (\mu_{s + 1/n}^{(n)*} - \mu_s^{(n)*}) = \frac{1}{nt - 1} (\mu_t^{(n)*} - \mu_{1/n}^{(n)*}) \leq \frac{1}{nt - 1}.
    \end{equation}
    Using this inequality, we can apply the squeeze theorem and compute $\mu_\infty^{(n)*}$ as follows: for $t \in\{T, T+1, \dots\}$, we have
   \begin{align*}
    0\leq \frac{\mu_t^{(n)*} - \mu_{t - 1/n}^{(n)*}}{\sqrt{(\eta_t^{(n)})^2 + (\eta_{t - 1/n}^{(n)})^2}}\leq \underbrace{(C/2)^{1/2} \frac{t^{1 - \epsilon/2}}{nt - 1}}_{\to 0 \quad \text{as $t \to \infty$}} \\
    \implies \mu_\infty^{(n)*}=c-\Phi\left(\lim_{t\to\infty}\frac{\mu_t^{(n)*} - \mu_{t - 1/n}^{(n)*}}{\sqrt{(\eta_t^{(n)})^2 + (\eta_{t - 1/n}^{(n)})^2}}\right)=c-(1/2)
    \end{align*}
    Note that since the sequence $(\mu_t^{(n)*})_{t \in \mathcal{T}^{(n)}}$ converges, the subsequence $(\mu_t^{(n)*})_{t \in \mathcal{T}}$ also converges to the same limit $\mu_\infty^{(n)*}$.

    \underline{\textbf{Part 3 (ii).}} By definition, $\eta^{-2}_t = \Omega(t^{2 + \epsilon})$ means that there exist $C$ and $T$ such that $\eta_t^{-2} \geq C t^{2 + \epsilon}$ for all $t \in \mathcal{T}$ and $t \geq T$. Since $t^{2 + \epsilon}$ is strictly convex, it must be that $(\eta_t^{(n)})^{-2} \geq C t^{2 + \epsilon}$ for all $t \in \mathcal{T}^{(n)}$ and $t \geq T$.

    To invoke Theorem~\ref{thrm:char}, consider an auxiliary environment in $\mathcal{T}$ such that $\tilde\eta_s := \eta_{s/n}^{(n)}$ for $s \in \mathcal{T}$. Then, for all $s \in \mathcal{T}$ and $s \geq nT =: \tilde T$, we have
    \begin{align*}
        (\tilde\eta_s)^{-2} &= (\eta_{s/n}^{(n)})^{-2} \geq C \left( \frac{s}{n} \right)^{2 + \epsilon} \tag{$\frac{s}{n} \in \mathcal{T}^{(n)}$ and $\frac{s}{n} \geq T$} \\
        &= \underbrace{Cn^{-(2 + \epsilon)}}_{=: \tilde C} s^{2 + \epsilon},
    \end{align*}
    which implies $(\tilde\eta_s)^{-2} = \Omega (s^{2 + \epsilon})$. Hence, Theorem~\ref{thrm:char} (ii) implies $\lambda_\infty^{(n)}(\theta)=NRD(\theta)$.
    % , as the limit belief thresholds induced by $(\tilde\eta_s)_{s \in \mathcal{T}}$ and $(\eta_t^{(n)})_{t \in \mathcal{T}^{(n)}}$ coincide.
\end{proof}

% \clearpage
\section{Signals about past play}\label{appendix:pastplay}
In the main text, we claimed that our model is equivalent to one in which players receive noisy signals of past play. We now formalize this claim by considering a single modification to our model: at each time $t$, each player $i \in [0,1]$ observes the independent signals 
\[
x_{it} \sim N\Big(\theta,\sigma^2_t \Big) \quad \text{and} \quad y_{it} \sim N\Big(\Phi^{-1}(\lambda_{t-1}),\tau^2_t\Big) \text{ for $t > 1$}
\]
This formulation follows that of \cite{dasgupta2007coordination} and \cite{trevino2020informational}. 

\begin{proposition}
$(\lambda_t)_t$ is consistent with some $\Sigma$ under our main model if and only if $(\lambda_t)_t$ is consistent with some $(\tau_t^2,\sigma_t^2)_{t\geq 1}$ under the Appendix B model.
\end{proposition}

\begin{proof}
Notice that the environment in the main text can be viewed as the case in which $\tau^2_t = +\infty$ for all $t \geq 0$. Also observe that the case in which $\sigma^2_1 < +\infty$, $\sigma^2_t = +\infty$ for all $t  > 1$ can be viewed as a social learning environment: at the start of the game, each player receives an independent signal of the state and, over the course of the game, the path of play also aggregates information about the state. 

We now show that the dynamics of the game induced by the path $(\tau^2_t,\sigma^2_t)_{t \geq 1}$ can be replicated by some alternate path $(\sigma'{^2}_t)_{t \geq 1}$ in our main model in which past play is unobserved. Let $\eta_t^{-2}$ and $\eta_t'^{-2}$ be the posterior precisions induced by $(\tau^2_t,\sigma^2_t)_{t \geq 1}$ and $\Sigma'$, respectively.

As before, let $\eta^2_t$ be the posterior variance at time $t$ 
for $t \geq 2$ with the initial conditions $\mu_{i1} = x_{i1}$ and $\eta_1^2=\sigma_1^2$. Then, the law of motion is unchanged from the main model and given by Lemma \ref{lem: lom}.

We can show the law of motion of $\mu_{it}$ and $\eta_t^2$ and \eqref{eq: lam LoM} by induction. Define $I_{is}=(x_{is},y_{is})_{s\leq t}$. At $t=1$, agent $i$'s posterior is given by $\theta|I_{i1}\sim N(x_{i1},\sigma_1^2)$, which implies $\mu_{i1} = x_{i1}$ and $\eta_1^2 = \sigma_1^2$. Since $\lambda_0$ is common knowledge, she takes $a_{i1}=1$ if and only if $\mu_{i1}\geq c-\lambda_0=\mu_1^*$. Since $\mu_{i1}|\theta\sim N(\theta,\eta_1^2)$, $\lambda_1$ is given by 
    \begin{align*}
        \lambda_1=\text{Pr}(x_{i1}\geq\mu_1^*)=\Phi\left(\frac{\theta-\mu_1^*}{\sigma}\right).
    \end{align*}
    
Now suppose that the statement holds for $t-1$. By assumption, we have $\Phi^{-1}(\lambda_{t-1})=(\theta-\mu_{t-1}^*)/\eta_{t-1}$, which implies
    \begin{align*}
        \mu_{it-1}\,|\,\theta&\sim N(\theta,\eta_{t-1}^2),\\
        x_{it}\,|\,\theta&\sim N(\theta,\sigma_t^2),\\
        (\eta_{t-1}y_{it}+\mu_{t-1}^*)\,|\,\theta&\sim N(\theta,\tau_t^2\eta_{t-1}^2).
    \end{align*}
Since they are independent of each other conditional on $\theta$, Bayesian updating yields
$\eta_1^{-2}=\sigma_1^{-2}$, $\mu_{i1}=x_{i1}$, and for each $t\geq 2$,
\begin{align*}
    \eta_t^{-2}=\eta_{t-1}^{-2}(1+\tau_t^{-2})+\sigma_t^{-2} \\
    \mu_{it}=\sum_{s=1}^t \left(\frac{\eta_{s-1}^{-2}\tau_s^{-2}}{\eta_t^{-2}}\tilde{y}_{is}+\frac{\sigma_s^{-2}}{\eta_t^{-2}} x_{is} \right)
\end{align*}
where $\tau_1^{-2}:=0$. This implies:
\[
\eta_t^{-2}=\sum_{s=1}^t \sigma_s^{-2}\prod_{r>s}^t (1+\tau_r^{-2})
\]

% the law of following expressions
%     \begin{align*}
%         \mu_{it}&=\frac{\sigma_t^2\tau_t^2\mu_{it-1}+\eta_{t-1}^2\tau_t^2x_{it}+\sigma_t^2(\eta_{t-1}y_{it}+\mu_{t-1}^*)}{\sigma_t^2+\tau_t^2(\sigma_t^2+\eta_{t-1}^2)}\\
%         \eta_t^2&=\frac{\sigma_t^2\tau_t^2\eta_{t-1}^2}{\sigma_t^2+\tau_t^2(\sigma_t^2+\eta_{t-1}^2)}
%     \end{align*}   
Since \eqref{eq: mu star} remains unchanged from Lemma~\ref{lem: lom}, we obtain \eqref{eq: lam LoM} by the proof of Lemma~\ref{lem: lom}. Then, we can find a sequence $(\sigma'_t)_t$ such that the dynamics are equivalent:
\[
\sigma_t'^{-2}:=\eta_t^{-2}-\eta_{t-1}^{-2}
\]
\end{proof}

% \setstretch{1}
% \setlength{\bibsep}{0pt}
% \bibliography{WorksCited}

\section{Finite Players}\label{appendix:finiteplayers}
Consider the following game, which we denote as the \textit{$N$-player model}, in contrast to the main text's \textit{continuum model}. There are $N$ players, indexed by $i \in [N]$. Actions, timing, and information are the same as in the main text.\footnote{Define the random variables $(x_{it})_{i\geq 1,t\geq 1}$ on the same probability space.} Agent $i$'s payoffs at time $t$ are:
\[
u(a_{it},\lambda_{N,-i,t-1},\theta):=\begin{cases} 
\theta+\lambda_{N,-i,t-1} & \text{if } a_{it}=1 \\
c & \text{if } a_{it}=0 \\
\end{cases}
\]
where
\[
\lambda_{N,-i,t-1}=\frac{1}{N-1}\sum_{j\in [N]\setminus\{i\}} \boldsymbol{1}[a_{j,t-1}=1]
\]
is the empirical proportion of non-$i$ agents who took the risky action at time $t-1$ for $t>1$, and $\lambda_{N,-i,t-1}:=\lambda_0$ is common knowledge. Let
\[
\lambda_{N,t}:=\frac{1}{N} \sum_{i=1}^N \boldsymbol{1}[a_{it}=1]
\]
be the empirical proportion of all agents who attack at time $t$. 

The purpose of Appendix C is to show that the dynamics of the $N$-player model converge to the dynamics of the continuum model as $N$ becomes large. Let us briefly discuss this exercise. With a finite number of players, there is randomness in the empirical distribution of beliefs at time-$t$. This, in turn, translates into randomness in the path of aggregate play. Theorem \ref{thrm:approx} shows that in an environment with many players, this randomness washes out.

To show Theorem \ref{thrm:approx}, we begin by analyzing the dynamics of belief thresholds in the $N$-player model.

\begin{lemma}
At each time $t\geq 1$, each agent $i \in [N]$ takes the risky action if and only if $\mu_{it}\geq \mu_t^*$, where $(\mu_t^*)_{t\geq 1}$ is the same sequence of thresholds as in the main model.
\end{lemma}

\begin{proof}
\underline{Base case}: At time $t=1$, each agent $i$ takes the risky action iff $\mu_{i1}=x_{i1}\geq c-\lambda_0=:\mu_1^*$. 

% Hence,
% \[
% \lambda_{N,1}=\frac{1}{N}\sum_{i=1}^N \boldsymbol{1}[\mu_{i1}\geq \mu_1^*]
% \]
\underline{Inductive step}: Fix $t>1$ and assume that at time $t-1$, each agent $i$ takes the risky action iff $\mu_{i,t-1}\geq \mu_{t-1}^*$. At time $t$, each agent $i$ takes the risky action iff
\begin{align*}
    \mathbb{E}_{it}\Big[\theta+\lambda_{N,i,t-1}\Big]=\mu_{it}+\frac{1}{N-1}\sum_{j\neq i} \mathbb{P}_{it}\Big(\mu_{j,t-1}\geq \mu_{t-1}^*\Big)\geq c 
\end{align*}
Note that at time $t$, agent $i$ believes $\theta \sim N(\mu_{it},\eta_t^2)$ and for each $j\neq i$, $\mu_{j,t-1}|\theta \sim N(\theta,\eta_{t-1}^2)$. Hence, at time $t$ agent $i$ believes 
\[
\mu_{j,t-1} \sim N(\mu_{it},\eta_{t-1}^2+\eta_t^2)
\]
which implies she takes the risky action at time $t$ iff
\[
\mu_{it}+\Phi\left(\frac{\mu_{it}-\mu_{t-1}^*}{\sqrt{\eta_{t-1}^2+\eta_t^2}}\right)\geq c \iff \mu_{it}\geq \mu_t^*
\]
as desired.
\end{proof}

With this lemma in hand, observe that
\[
\lambda_{N,t}=1-F_{N,t}(\mu_t^*) \quad \text{and} \quad \lambda_t(\theta)=1-F_{\theta,t}(\mu_t^*)
\]
where $F_{N,t}$ is the empirical CDF of $N$ draws of $\mu_{it}$ and $F_{\theta,t}$ is the CDF of $\mu_{it} \sim N(\theta,\eta_t^2)$.

\begin{theorem}\label{thrm:approx}
For any $N\geq 1$, $(\Sigma,\lambda_0)$, and $\theta \in \mathbb{R}$,
\[
|\lambda_{N,\infty}-\lambda_\infty(\theta)|=0 \quad \mathbb{P}_\theta-\text{a.s.}
\]
% \begin{enumerate}
%     % \item[(i)] For any $T\geq 1$,
%     % \[
%     % \sup_{t \in [T]} |\lambda_{N,t}-\lambda_t(\theta)| \to 0 \quad N\to\infty \ \ \mathbb{P}_\theta-\text{a.s.} 
%     % \]
%     \item[(ii)]
    
% \end{enumerate}
\end{theorem}

\begin{proof}
% \underline{Part (i)}: Fix any $T\geq 1$ and $t \in [T]$. Note that
% \begin{align*}
%     |\lambda_{N,t}-\lambda_t(\theta)|=|F_{N,t}(\mu_t^*)-F_{\theta,t}(\mu_t^*)| \\
%     \leq \sup_{a \in \mathbb{R}} |F_{N,t}(a)-F_{\theta,t}(a)| \to 0 \quad \text{a.s. as } N \to \infty
% \end{align*}
% where the last limit follows from the Glivenko--Cantelli theorem. Since this holds for each $t \in [T]$, the conclusion follows from the union bound.

% \underline{Part (ii)}
For each $N\geq 1$ and each $i \in [N]$, note that $\mu_{i,\infty}=\theta$ $\mathbb{P}_\theta$-a.s. by Doob's theorem. Hence, $\lambda_{N,\infty}=\lambda_\infty(\theta)$ $\mathbb{P}_\theta$-a.s.
\end{proof}

\section{Alternative Interpretations of Model}
\label{app:microf}
\textbf{Overlapping generations model.} There is an initial unit mass of agents with proportion $\lambda_0 \in (0,1)$ who take the risky action, all of whom share a common improper uniform prior about $\theta$. At each time $t\geq 1$, a unit mass of children is born (generation $t$), each of whom inherits her prior about $\theta$ from her parent's (generation $t-1$) belief. Each child $i$ receives a signal $x_{it} \sim N(\theta,\sigma_t^2)$, updates her belief about $\theta$, chooses an action, and gets uniformly matched to play the static game below with a member of generation $t-1$. Finally, child $i$ becomes parent $i$ and continues to take her time-$t$ action whenever matched with an agent from generation $t+1$. We interpret this as due to switching costs---since agents' actions when playing against the future generation are not payoff-relevant, switching costs ensure this.

\textbf{Two-player interpretation.} Consider the set of static symmetric two-player games, indexed by the state $\theta \in \mathbb{R}$, whose payoff matrix is
\[
\begin{blockarray}{ccc}
\ & A & N \\
\begin{block}{c(cc)}
  A & (\theta+1,\theta+1) & (\theta,c) \\
  N & (c,\theta) & (c,c) \\
\end{block}
\end{blockarray}
\]
Note that for $\theta \in (c-1,c)$, there are two pure strategy equilibria: $(A,A)$ and $(N,N)$. Applying \cite*{harsanyi1988general}'s risk dominance criterion selects $(A,A)$ if and only if $\theta \geq c-(1/2)$.

\cite*{carlsson1993global} propose a static global games framework in which two players play the incomplete information game with state-dependent matrix from above.\footnote{Our static two-player game corresponds to Figure~1 of \cite*{carlsson1993global} with $x = 4(c - \theta)$, $\alpha_1 = \alpha_2 = N$, and $\beta_1 = \beta_2 = A$.} Consider the case where both players share a common improper uniform prior, and each receive iid Gaussian signals distributed $N(\theta,\sigma^2)$. \cite*{carlsson1993global} then show that as $\sigma \downarrow 0$, the risk dominant equilibrium is selected.

In this section, we show a tight connection between the above static global games setting and our model where $\Sigma=(\sigma^2,\infty,\infty,\ldots)$, which we study in Example \ref{ex:brd}. In this case, $(\mu_t^*)_t$ evolves according to the law of motion
\[
\mu_t^*+\Phi\left(\frac{\mu_t^*-\mu_{t-1}^*}{\sqrt{2}\sigma}\right)=c
\]
with $\mu_1^*=c-\lambda_0$. As we have shown, this equation has a unique fixed point at $c-(1/2)$ and hence we have global convergence $\mu_t^* \to c-(1/2)$ for any $\lambda_0 \in (0,1)$. Furthermore, the entire path of $(\mu_t^*)_t$ corresponds to the iterated best-response dynamics with symmetric switching strategies in \cite*{carlsson1993global}'s two-player setting with the informational environment from above, in which $\mu_\infty^*$ is the unique rationalizable switching strategy. To see this, let $\mathbb{P}_x(\theta,x')$ denote the joint posterior over the state $\theta$ and my opponent's signal $x'$, conditional on my signal $x$.

The key observation is that, if my opponent follows a switching strategy with cutoff $\mu'$, my expected payoff from $A$ given $x$ is
\begin{align*}
    \mathbb{P}_x(x'\geq \mu')\Big( \mathbb{E}_x[\theta|x'\geq \mu']+1\Big)+\mathbb{P}_x(x'<\mu')\mathbb{E}_x[\theta|x'<\mu'] \\
    =\mathbb{E}_x[\theta]+\mathbb{P}_x(x'\geq \mu')=x+\int_\theta \mathbb{P}(x'\geq \mu'|\theta,x) d\mathbb{P}_x(\theta)=x+\mathbb{E}_x[\lambda(\theta)]
\end{align*}
which proves that if my opponent follows a switching strategy with cutoff $\mu_{t-1}^*$, my best response is to follow a switching strategy with cutoff $\mu_t^*$, as desired.

An analogous argument allows us to extend the interpretation of our model dynamics as two-player iterated joint learning and best-response dynamics by considering our model for arbitrary $\Sigma$. In this case, note that if my opponent receives signal distributed $N(\theta,\eta_{t-1}^2)$ (corresponding to the time-$(t-1)$ distribution of posterior means, conditional on the state) and my belief is about the state is $N(\mu_{it},\eta_t^2)$, then my (state-unconditional) belief about my opponent's signal is
\[
x' \sim N(\mu_{it},\eta_{t-1}^2+\eta_t^2) \implies \mathbb{P}_{it}(x'\geq \mu')=\Phi\left(\frac{\mu_{it}-\mu'}{\sqrt{\eta_{t-1}^2+\eta_t^2}}\right)
\]
and hence our general setting is also consistent with iterated joint learning and best-response dynamics. This can be interpreted as best-response dynamics with bounded rationality, since at each time each player is only able to take into account one extra signal.

\section{Payoffs generalizations} \label{appendix:payoffextension}
In the main text, we assumed that the payoff of taking action $1$ is the sum of the state and the measure of agents who played action $1$ in the previous period. Now we relax this assumption and suppose agent $i$'s payoffs at time $t$ are:
\[
        u(a_{it},\lambda_{t-1},\theta):=\begin{cases} 
        a\theta+ b \lambda_{t-1} & \text{if } a_{it}=1 \\
        c & \text{if } a_{it}=0 \\
        \end{cases}
\]
with $a \in (0,1)$ and $b \in (0,1)$. The parameter of $b / a$ can be interpreted as the extent of intertemporal strategic complementarity. To simplify the notations, we define the followings:
\begin{align*}
    \tilde\theta (a, b) &:= \frac{a}{b} \theta \\
    \tilde c (a, b) &:= \frac{c}{b} \\
    \tilde\sigma_t^2 (a, b) &:= \left( \frac{a}{b} \right)^2 \sigma_t^2 \\
    \tilde\eta_t^2 (a, b) &:= \left( \frac{a}{b} \right)^2 \eta_t^2.
\end{align*}
In this environment, we define the risk dominant action at state $\theta$ as
\[
RD(\theta) := \mathbbm{1}\left[\theta \geq \frac{c}{a}-\frac{1}{2}\frac{b}{a} \right] = \mathbbm{1}[\tilde{\theta}(a, b) \geq \tilde{c}(a, b)-(1/2)]
\]
i.e., the action that one would take if she knew the state and conjectured that half of all agents take the risky action. Analogously, we define the non-risk dominant action as $NRD(\theta):=1-RD(\theta)$.

\begin{proposition}
    All the results in the main text hold by replacing $c$, $\sigma_t^2$, $\eta_t^2$, and $\bar\beta = 1+|c-(1/2)-\theta|/2$ with $\tilde c (a, b)$, $\tilde\sigma_t^2 (a, b)$, $\tilde\eta_t^2 (a, b)$, and $\tilde{\beta} := 1+\big|\tilde c (a, b)-(1/2)-\tilde\theta (a, b)\big|/2$, respectively.
\end{proposition}

\begin{proof}
    We show \eqref{eq: mu star} and \eqref{eq: lam LoM} by induction:
    \begin{gather*}
        \mu_t^*+\Phi\left(\frac{\mu_t^*-\mu_{t-1}^*}{\sqrt{\tilde\eta_{t-1}^2 (a, b)+\tilde\eta_t^2 (a, b)}}\right)=\tilde c (a, b) \\
        \lambda_t(\theta) = \mathbbm{1}\left\{ \tilde\theta(a, b) \geq \mu_t^* \right\} = \Phi\left(\frac{\tilde\theta (a, b) - \mu_t^*}{\tilde\eta_t (a, b)}\right)
    \end{gather*}
    At $t=1$, agent $i$'s posterior belief is given by $\theta|I_{i1}\sim N(x_{i1},\sigma_1^2)$, which implies $\tilde\theta (a, b) |I_{i1} \sim N(\tilde\mu_{i1},\tilde\sigma_1^2 (a, b))$ with $\tilde\mu_{i1} := \mathbb{E}_{i1}[\tilde\theta (a, b)] = a x_{i1}/ b =: \tilde x_{i1} $. Since $\lambda_0$ is common knowledge, she takes $a_{i1}=1$ if and only if $\tilde\mu_{i1}  \geq \tilde c-\lambda_0=\mu_1^*$. Since $\tilde\mu_{i1}|\tilde\theta (a, b)\sim N(\tilde\theta (a, b),\tilde\sigma_1^2 (a, b))$, $\lambda_1$ is given by 
    \begin{align*}
        \lambda_1=\mathbb{P}_\theta(\tilde\mu_{i1}\geq\mu_1^*)=1-\Phi\left(\frac{\mu_1^*-\tilde\theta (a, b)}{\tilde\sigma_1 (a, b)}\right)=\Phi\left(\frac{\tilde\theta (a, b)-\mu_1^*}{\tilde\sigma_1 (a, b)}\right).
    \end{align*}
    Hence, \eqref{eq: lam LoM} holds for $t=1$.
    
    Now, for $t>1$, suppose that $a_{j,t-1}=1$ if and only if $\tilde\mu_{j,t-1}\geq \mu_{t-1}^*$. By standard Gaussian-Gaussian Bayesian updating, we have $\tilde\theta (a, b)|I_{it} \sim N(\tilde\mu_{it},\tilde\eta_t^2(a, b))$, where
    \begin{align*}
        \tilde\mu_{it}&=\frac{\tilde\eta_{t-1}^{-2}(a, b)}{\tilde\eta_{t-1}^{-2}(a, b)+\tilde\sigma_t^{-2}(a, b)}\tilde\mu_{i,t-1}+\frac{\tilde\sigma_t^{-2}(a, b)}{\tilde\eta_{t-1}^{-2}(a, b)+\tilde\sigma_t^{-2}(a, b)} \tilde x_{it}=\tilde\eta_t^2(a, b)\sum_{s=1}^t \tilde\sigma_s^{-2}(a, b) \tilde x_{is} \\
        \tilde\eta_t^2 (a, b) &=\frac{\tilde\eta_{t-1}^2(a, b)}{1+\tilde\sigma_t^{-2}(a, b)\tilde\eta_{t-1}^2(a, b)}=\left(\sum_{s=1}^t \tilde\sigma_s^{-2}(a, b)\right)^{-1},
    \end{align*}
    where $\tilde x_{is} := a x_{is}/b$. Then, by an analogous argument as in the proof of Lemma~\ref{lem: lom}, \eqref{eq: mu star} and \eqref{eq: lam LoM} hold for $t$. Once we obtain \eqref{eq: mu star} and \eqref{eq: lam LoM}, the rest of the arguments carry through by replacing $c$, $\sigma_t^2$, $\eta_t^2$, and $\bar\beta = 1+|c-(1/2)-\theta|/2$ with $\tilde c (a, b)$, $\tilde\sigma_t^2 (a, b)$, $\tilde\eta_t^2 (a, b)$, and $\tilde{\beta} := 1+\big|\tilde c (a, b)-(1/2)-\tilde\theta (a, b)\big|/2$, respectively.
\end{proof}

In particular, since learning rates are unaffected by scaling, the tight connection between learning rates and limit play in Theorem~\ref{thrm:char} is robust to the extent of coordination:

\begin{corollary}
Theorem~\ref{thrm:char} holds exactly as stated in the main text, \textit{without} replacing the above expressions with their $a,b$-counterparts.
\end{corollary}

\begin{proof}
Note that $\eta_t^{-2}=O(t^p)$ if and only if $\tilde{\eta}_t^{-2}(a, b)=O(t^p)$, and that $\eta_t^{-2}=\Omega(t^q)$ if and only if $\tilde{\eta}_t^{-2}(a, b)=\Omega(t^q)$. Hence, Theorem~\ref{thrm:char} holds without replacing the above expressions with their $a,b$-counterparts.
\end{proof}

\subsection{Heterogeneous Unknown Initial Shocks}
\label{app:het-unknown}
\begin{lemma}
Consider the model and payoffs of the main text, with one generalization: $\theta_{i1}:=\theta+\lambda_{i0}(\theta)$, where $\{\lambda_{i0}(\cdot)\}_{i \in \mathcal{I}}$ are bounded, continuous, and increasing in $\theta$ and distributed independently from agents' information. For all $t\geq 2$, players use identical cutoff strategies.
\end{lemma}

\begin{proof}
We show this by induction. The inductive step is exactly the same as the inductive step of Lemma \ref{lem: lom}. For the base case of $t=2$, observe that $a_{i2}=1$ iff 
\[
\mathbb{E}_{i2}[\theta+\lambda_1(\theta)]\geq c
\]
where $i$'s belief is $N(\mu_{i2},\eta_2^2)$. Since $\mu\geq \mu'$ implies $N(\mu,\eta_2^2) \geq_{FOSD} N(\mu',\eta_2^2)$, it suffices to show that $\lambda_1(\theta)$ is increasing. Note that at $t=1$, $a_{i1}=1$ iff 
\[
\mathbb{E}_{i1}[\theta+\lambda_{i0}(\theta)]\geq c \iff \mu_{i1}\geq \mu_{i1}^*
\]
by an analogous argument as above and in Lemma \ref{lem: lom}, and since $\lambda_{i0}(\theta)$ is increasing. Since $\mu_{i1}$ and $\mu_{i1}^*$ are independent random variables on the probability space constructed by \cite*{sun2009individual},
\[
\lambda_1(\theta)=\mathbb{P}\Big(\mu_{i1}\geq \mu_{i1}^*\Big)=\mathbb{E}_{\mu_{i1}^*} \Phi\left(\frac{\theta-\mu_{i1}^*}{\eta_t}\right)
\]
which is increasing in $\theta$, as desired.
\end{proof}

Here, we relax common knowledge about $\lambda_0$ in two ways. First, we allow for the initial shock to depend on the underlying fundamental, such that no player knows the value of the initial shock (and therefore must form beliefs about it). Second, we allow for the initial shock to be heterogeneous across players $i$. While this induces heterogeneous cutoffs at $t=1$, this Lemma shows that for all $t\geq 2$, players revert to using identical cutoff strategies.

\end{document}